\documentclass[10pt]{article}

\usepackage[utf8]{inputenc}
\usepackage{amsmath,amsfonts,amsthm,amssymb,upgreek}

\usepackage{color}

\usepackage[dvipdfm,a4paper,left=30mm,right=30mm,top=35mm,bottom=35mm]{geometry}
\usepackage{dsfont,paralist, enumitem}
\usepackage{lscape,soul}

\usepackage{lineno}

\newtheorem{theorem}{Theorem}
\newtheorem{proposition}{Proposition}
\newtheorem{remark}{Remark}
\newtheorem{corollary}{Corollary}
\usepackage[caption=false]{subfig}

\newcommand{\R}{\mathbb{R}}
\newcommand{\g}{\gamma}
\newcommand{\pd}{\mathbb{\partial}}
\newcommand{\Om}{\Omega}

\newcommand{\eps}{\varepsilon}
\def\d{\, \text{d}}

\newcommand{\ls}[1]{\textcolor{cyan}{#1}}

\let\tau\uptau
\usepackage{subfig}
\usepackage{graphicx}

\title{How environment affects active particle swarms: a case study}
\author{Pierre Degond$^{(1)}$, Angelika Manhart$^{(2)}$, Sara Merino-Aceituno$^{(3)}$, \\Diane Peurichard$^{(4)}$, Lorenzo Sala$^{(5)}$}  
\date{\today}

\begin{document}

\maketitle

\begin{center}
1- Institut de Mathématiques de Toulouse ; UMR5219 \\
Université de Toulouse ; CNRS \\
UPS, F-31062 Toulouse Cedex 9, France\\
email: pierre.degond@math.univ-toulouse.fr
\end{center}

\begin{center}
2- Mathematics Department, 
University College London,\\
25 Gordon Street, London, UK\\
email: a.manhart@ucl.ac.uk
\end{center}

\begin{center}
3- Faculty of Mathematics,
University of Vienna,\\
Oskar-Morgenstern-Platz 1,
1090 Vienna, Austria\\
email: sara.merino@univie.ac.at
\end{center}

\begin{center}
4- Inria, Laboratoire Jacques-Louis Lions, \\
Sorbonne Université, CNRS, Université de Paris\\
4, Place Jussieu,
F75252 Paris Cedex 05\\
email: diane.a.peurichard@inria.fr
\end{center}

\begin{center}
5- INRIA Saclay Ile-de-France\\
1 rue Honoré d’Estienne d’Orves\\
91120 Palaiseau,
FRANCE\\
email: lorenzo.sala@inria.fr
\end{center}

\begin{abstract}
We investigate the collective motion of self-propelled agents in an environment filled with obstacles that are tethered to fixed positions via springs. The active particles are able to modify the environment by moving the obstacles through repulsion forces. This creates feedback interactions between the particles and the obstacles from which a breadth of patterns emerges (trails, band, clusters, honey-comb structures,...). We will focus on a discrete model first introduced in \cite{Aceves2020} and derived into a continuum PDE model. 
As a first major novelty, we perform an in-depth investigation of pattern formation of the discrete and continuum models in 2D: we provide phase-diagrams and determine the key mechanisms for bifurcations to happen using linear stability analysis. As a result, we discover that the agent-agent repulsion, the agent-obstacle repulsion and the obstacle's spring stiffness are the key forces in the appearance of patterns, while alignment forces between the particles play a secondary role. The second major novelty lies in the development of an innovative methodology to compare discrete and continuum models that we apply here to perform an in-depth analysis of the agreement between the discrete and continuum models.
\end{abstract}

\clearpage
\tableofcontents
\clearpage

\section{Introduction}
Understanding how patterns in collective motion arise from local interactions between individuals is an exciting and challenging endeavour that has drawn the attention of the scientific community \cite{Ben-Jacob2000,Bernoff2016,Cavagna2010,Boissard2013,degonddimarco2015,helbing2005self,noselli2019swimming,Schoeller2018}. 
In many scenarios the environment plays a key role in the emergence of collective motion and of the resulting patterns \cite{chepizhko2013optimal,kamal2018enhanced,Lo2000,majmudar2012experiments,park2008enhanced,tung2017fluid}. 
Examples are evacuation dynamics in the presence of obstacles~\cite{Feliciani2016,helbing2005self,maury2010macroscopic}, sperm dynamics in the seminal fluid~\cite{Degond2019,tung2017fluid}, swirl of fish under the presence of predators~\cite{carrillo2010particle}, cells moving in a space filled with fibers~\cite{Lo2000} or over a substrate~\cite{noselli2019swimming},... 
\newline
In particular, we are interested by feedback interactions between self-propelled agents and their environment that they are able to modify. This happens, for example, {(i)} in the formation of paths in grass-land by active walkers~\cite{Helbing1997,Lam1995},(ii) in the modification of the extra-cellular matrix (fibers) by migratory cells~\cite{Baricos1995}, or {(iii)} in ant trail formation due to ant pheromone deposition \cite{Boissard2013}.  In this paper, we will focus on the model introduced in~\cite{Aceves2020} where collective motion happens in an environment filled with movable obstacles that are tethered to a fixed point via a spring. 
The authors in~\cite{Aceves2020} showed that a variety of patterns are generated due to the feedback interactions between the obstacles and the self-propelled agents. Indeed, the capacity of the agents to modify their environment (\textit{i.e.}, to modify the position of the obstacles) is key for patterns to form. 

\begin{figure}
    \centering
    \includegraphics[width=\textwidth]{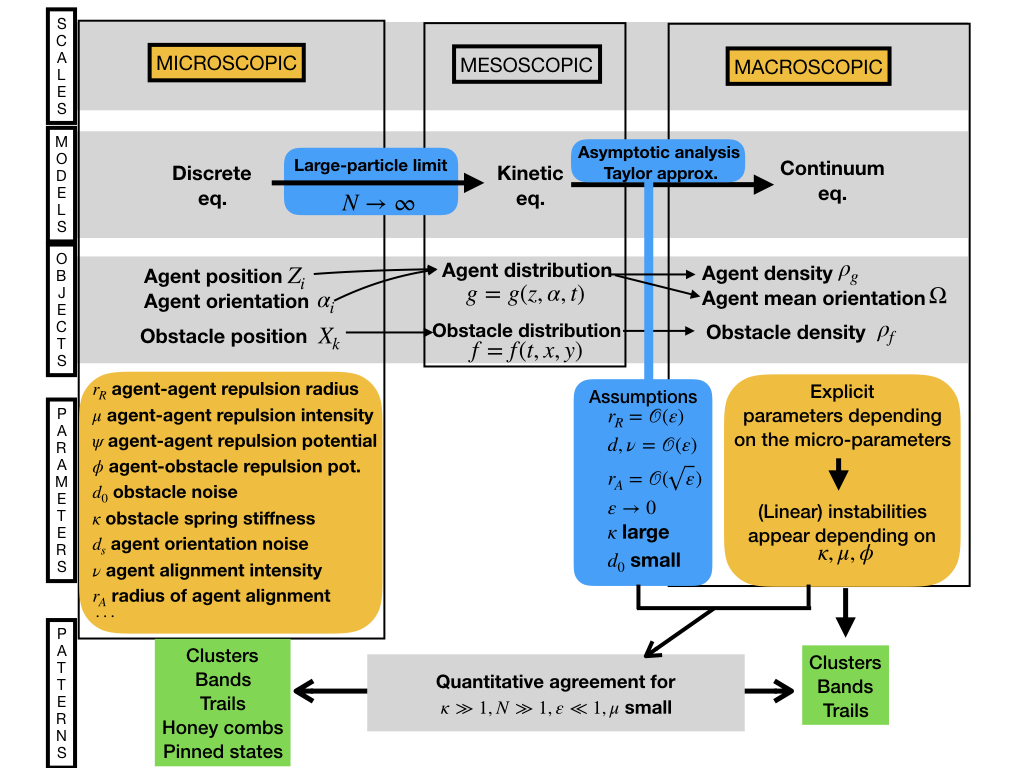}
    \caption{Overview of the paper. It includes a summary of the scales, the models and the objects considered in this paper and introduced in \cite{Aceves2020} (first three grey lines). The blue boxes indicate the derivation of the different models and derivation assumptions. The main contributions in the paper appear in the last row corresponding to `patterns' (at the discrete and continuum level and their correspondence)  and the linear stability analysis (bottom right yellow box).}
    \label{Fig:project_diagram}
\end{figure}

Figure~\ref{Fig:project_diagram} offers an overview of the ideas and messages of this paper. We will consider mostly two scales (marked in yellow). The reason for this is that understanding the emergent properties of collective dynamics  requires to establish a link between the agent's interactions and the continuum dynamics that emerge{s} at scales much larger than the size of the individual agents. 
As a consequence, it is natural to consider two different scales to investigate collective motion: a microscopic scale where the discrete dynamics of the agents can be described, and a macroscopic scale where the average/continuum behaviour of the large ensemble can be observed.  

From a modelling perspective, it is natural to consider the microscopic scale, where individual-based models can describe individual-agent behaviour and their interactions. In the left column of Figure~\ref{Fig:project_diagram} we present key features of the individual-based model introduced in~\cite{Aceves2020}.  
The model assumes that agents move trying to avoid obstacles via a repulsion force. 
Agents interact with each other following Vicsek-type dynamics \cite{Manhart2018,Ha2008,jiang2016hydrodynamic,Vicsek1995}, \textit{i.e.} they move at a constant speed trying to align their orientation of motion with the one of their neighbours, up to some noise, while repelling each other at short distances.  
The discrete system gives the time-evolution of the position of the obstacles $(X_i)_{i=1,\hdots, N}$ tethered at fixed anchor points $(Y_i)_{i=1,\hdots, N}$ \emph{via} a spring and the position and orientation of the self-propelled agents $(Z_k, \alpha_k)_{k=1, \hdots, M}$, where $X_i, Y_i, Z_k\in \mathbb{R}^2$ and $\alpha_k$ is a unit vector (see Eq.~\eqref{eqn:IBMmodel_massless} for a full mathematical description of the system and Figure~\ref{Fig:project_diagram} for a list of the most relevant parameters). We will explore the variety of patterns that arise depending on the values of the model parameters.

However, the simulation of the discrete model becomes quickly computationally challenging for systems composed of millions of individuals. 
Therefore, for large-particle systems,  continuum models are to be preferred since they provide information on the average behaviour and are computationally less costly (right column of Figure \ref{Fig:project_diagram}). Moreover, continuum models are the appropriate framework for studying large scale patterns and carry out mathematical analyses like linear stability analysis. The drawback is that, from a modelling perspective, they are harder to justify than individual based models. For this reason, one would like to derive the continuum dynamics from the discrete ones: this derivation validates the continuum models and provides understanding on the emergence of large-scale patterns. 
At the same time, during this derivation process, due to averaging and asymptotic analysis, some information on the discrete system can be lost.

 This rigorous derivation is precisely one of the purposes of kinetic theory. Kinetic theory has been successfully applied to the study of models like the Vicsek model \cite{Manhart2018,Ha2008,jiang2016hydrodynamic,Vicsek1995} and the Cucker-Smale model \cite{AS2019,carrillo2010asymptotic,cucker2007mathematics}. 
 Tools from kinetic theory were applied in \cite{Aceves2020} to the discrete model described above, see second and third rows in Figure \ref{Fig:project_diagram}. 
 
 First, the authors derive the mean-field limit equation (large-particle limit $N,M\to \infty$ for both agents and obstacles). This equation corresponds to a Kolmogorov-Fokker-Plank equation for the time-evolution of the distribution of the agents $g=g(z,\alpha,t)$ at position $z\in \mathbb{R}^2$ and orientation $\alpha$; and the time-evolution of the distribution of the obstacles  $f=f(x,y,t)$ at position $x\in \mathbb{R}^2$ with anchor point at $y\in \mathbb{R}^2$. 

Then, from the kinetic equations for these distributions, the authors in \cite{Aceves2020} obtained continuum equations for the system under some asymptotic assumptions on the parameters (right blue boxes in Figure \ref{Fig:project_diagram}). 
In particular, it is assumed a high stiffness of the obstacle springs, strong local agent-agent repulsion and fast agent alignment. 
In this regime, it was shown in \cite{Aceves2020} that the obstacle density $\rho_f=\rho_f(x,t)$ becomes a non-local function of the agent density $\rho_g=\rho_g(x,t)$ and that the continuum model consists of a system of two non-linear non-local equations for $\rho_g$ and the local mean orientation of the agents $\Omega=\Omega(x,t)$, see Eqs. \eqref{eqmacro}.

\bigskip
The main objective of this article is to investigate the influence of the tethered obstacles in pattern formation using the discrete and continuum models {first introduced} in \cite{Aceves2020}. 
The main contributions of this paper are listed below:
\begin{itemize}
    \item We focus our study primarily on the continuum equations (which were analysed only in dimension one in \cite{Aceves2020}). Here we introduce 2D simulations of the continuum equations and an extensive phase diagram (Sec. \ref{MicMac}) that shows the appearance of patterns depending on the value of the parameters (green box in Fig. \ref{Fig:project_diagram}). {We carry out} a linear stability analysis in 2D around uniform states and validate {this analysis} by {comparing} its predictions with the numerical simulations of the discrete and continuum models (right yellow box in Fig. \ref{Fig:project_diagram}).
    \item We document in which parameter regime the continuum equations capture the discrete patterns (bottom grey box in Fig. \ref{Fig:project_diagram}). To this aim, we propose a method to compare discrete and continuum simulations. {This novel} method provides an indicator of the distance between different patterns.
    \item Lastly, we {also} expand and greatly systematize the parameter exploration of the discrete model {supported by} a phase diagram. As a consequence, we detect two new patterns with respect to reference \cite{Aceves2020}: honeycombs structures and pinned agents states (left green box in Fig. \ref{Fig:project_diagram}).
\end{itemize}

\paragraph{Organisation of the paper.}
The paper is organized as follows: we first describe the models (discrete and continuum), including the derivation assumptions of the continuum model. Then we simulate both systems to construct {two corresponding} phase diagrams based on different values of the parameters. Next, to {better} understand pattern formation as function of the model parameters, we perform a linear stability analysis of the continuum equations around uniform states and identify bifurcation parameters controlling the formation of patterns. Finally, an {innovative} method is proposed to compare discrete and continuum simulations, which is used to determine in which parameter regime the continuum equations are in good accordance with the discrete dynamics.
We conclude the paper with a discussion of the main results.


\section{Modeling}
\subsection{Discrete dynamics}\label{IBM}
 We consider as a starting point the model introduced in \cite{Aceves2020} for self-propelled particles undergoing collective motion in an environment filled with obstacles. 
 Obstacles are tethered to a given fixed anchor point through a Hookean spring. 
They are characterised by their positions $X_i(t)\in \R^2$ over time $t\geq 0$ and their anchor points $Y_i\in \R^2$ for $i = 1, 2, \ldots, N$, where $N$ is the total number of obstacles. The self-propelled particles  are characterised by their positions $Z_k(t)\in \R^2$ and orientations $\alpha_k(t)\in \mathbb{S}^{1}$ (unit circle) at time $t\geq 0$, $k = 1, 2, \ldots , M$, where $M$ is the total number of agents.
We assume that  obstacles and agents interact through a given potential, as explained next.

The evolution for the obstacles $(X_i(t), Y_i)_{i=1,\hdots,N}$ and the agents $(Z_k(t), \alpha_k(t))_{k=1,\hdots,M}$ over time is given by the following coupled system of stochastic differential equations: 
\begin{subequations}
\label{eqn:IBMmodel_massless}
\begin{align}
 \d X_i                                 =& -\frac{\kappa}{\eta}(X_i-Y_i)\d t -\frac{1}{\eta}\frac{1}{M}\sum_{ k = 1}^M \nabla \phi \left(X_i - Z_k\right) \d t+ \sqrt{ 2 d_o} \, \d B^{i}_t , \label{eq:IBM_obstacle}\\
  \d Z_k                                 =& u_0\alpha_k \d t - \frac{1}{\zeta}\frac{1}{N}\sum_{i = 1}^N \nabla \phi \left(Z_k - X_i\right)\d t -\frac{1}{\zeta}\frac{1}{M} \sum_{l \neq k}^M \nabla \psi \left(Z_k - Z_l\right)\d t,\label{eq:IBM_swim_positions}\\
  \d \alpha_k                         =& P_{ \alpha_k^\perp} \circ \bigg[ \nu \bar{\alpha}_k \d t + \sqrt{2 d_s} \, \d \tilde B^{k}_t \bigg]\label{eqn:IBMorient},
\end{align}
\end{subequations}
where the mean direction $\bar{\alpha}_k$ is defined via the mean flux $J_k$ as follows
\begin{align}
\label{eqn:meanDir}
\bar{ \alpha}_k = \frac{ J_k}{ | J_k|},\quad \text{ where } J_k =\mkern-18mu \sum_{\substack{j = 1 \\ |Z_k-Z_j|\leq r_A}}^M \mkern-18mu \alpha_j.
\end{align}
Eq. \eqref{eq:IBM_obstacle} gives the time-evolution for the obstacles' positions $X_i$. 
The first term on the right-hand side corresponds to the force generated by the Hookean spring anchored at position $Y_i$ with stiffness constant $\kappa>0$. 
The tether positions $Y_i$ are given and do not change over time. 
The terms $B^i$, $i=1,\hdots, N$ are independent Brownian motions that introduce noise in the dynamics with intensity $d_0>0$. This term accounts for fluctuations in the dynamics. 
Finally, the second term on the right-hand side of Eq. \eqref{eq:IBM_obstacle} is precisely the interaction force that couples the dynamics of the self-propelled agents with the ones of the obstacles. We assume that $\phi$ is an even and non-negative {interaction} potential. Typically we will assume $\phi$ to be a repulsive potential to model volume exclusion between obstacles and self-propelled particles.

Now, Eq. \eqref{eq:IBM_swim_positions} gives the time-evolution for the position of the self-propelled agents $Z_k$. The first term on the right-hand-side of \eqref{eq:IBM_swim_positions} expresses that agent $k$ moves in the orientation $\alpha_k$ at a fixed speed $u_0>0$. The second term is the force due to the {interaction} potential coupling the self-propelled agents and the obstacles, as we have seen before. Finally, the last term is a {repulsive} force between agents given by a potential $\psi$ which is assumed to be non-negative and even. This force is added to the model to prevent agents clustering at a single point in space and represents volume exclusion interactions between the agents \cite{degonddimarco2015}. 

The last equation \eqref{eqn:IBMorient} gives the time-evolution for the orientation of the agents and corresponds to the terms appearing in the Vicsek model \cite{Degond2008} which is a widely used model in collective motion. 
The right hand side of Eq. \eqref{eqn:IBMorient} is the sum of two competing forces: a force that tries to align the orientation of the self-propelled agents with the mean orientation of their neighbours and a noise term that opposes this alignment. The noise is given by $(\tilde B^k)_{k=1,\hdots, M}$ which are $M$ independent Brownian motions (also assumed to be independent from $B^i$, $i=1,\hdots, N$) and the intensity of this noise is given by the parameter $d_s>0$. The operator $P_{\alpha_k^\perp}$ represents the orthonormal projection onto $\alpha_k^\perp$ (where $\alpha_k^
\perp$ is a vector orthogonal to $\alpha_k$) and the symbol  $'\circ'$ indicates that the stochastic differential equation has to be understood in the Stratonovich sense \cite{hsu2002stochastic}. 
In particular, the projection ensures that, for all times where the dynamics are defined, $\alpha_k(t)$ remains on the sphere, \textit{i.e.}, $|\alpha_k|=1$. 
The alignment force is given by $P_{\alpha_k^\perp}\nu \bar \alpha_k$ where $\nu>0$ is a positive constant and $\bar \alpha_k$ is the average orientation of the neighbouring agents that are at distance $r_A>0$ from agent $k$, as computed in Eq. \eqref{eqn:meanDir}. Indeed, this term corresponds to an alignment force since it can be rewritten as
$$P_{\alpha_k^\perp}\nu \bar \alpha_k= \nu\nabla_{\alpha_k}(\alpha_k\cdot \bar \alpha_k),$$ 
where $\nabla_{\alpha_k}$ denotes the gradient on the sphere. 
Therefore, this term is a gradient flow that relaxes $\alpha_k$ towards the average orientation $\bar \alpha_k$ at speed $\nu>0$.

\bigskip
Finally, notice that the discrete model \eqref{eqn:IBMmodel_massless} consists of first order equations: the  model can be derived from second order equations in the overdamped (or inertialess) regime. This is the reason why the parameters $\eta>0$ and $\zeta>0$ appear in the system: $\eta$ corresponds to the obstacle friction and $\zeta$ to {the} agent friction. 
In an inertialess regime first-order equations give a good approximation of the dynamics and this regime appears in many biological applications, in particular involving micro-agents (like sperm cells) in highly viscous environments.

As we will see in later sections, the feedback interactions between agents and between agents and obstacles give rise to a variety of patterns depending on the value of the parameters.

\subsection{Continuum dynamics}\label{Derivation}
When the number of agents and obstacles becomes large, it is useful to derive equations that determine the average behaviour of the discrete system \eqref{eqn:IBMmodel_massless}. These `averaged' equations correspond to continuum equations, which were derived in \cite{Aceves2020} for the discrete system \eqref{eqn:IBMmodel_massless}. In this section we summarise the results from this reference.

\subsubsection{Main assumptions of the derivation}
\label{sec:assumptions}

The derivation of the continuum equations in \cite{Aceves2020} is done under the following set of assumptions:

\paragraph{(a) Large-particle-system assumption.}
The number of obstacles and agents are assumed to tend to infinity, \textit{i.e.}, $N\to \infty$, $M\to \infty$.\\
Under this assumption, the authors derived formally equations for the evolution of obstacles and agent density (kinetic equations). Then, some of the parameters of the kinetic equation are scaled by a small factor $\varepsilon \ll 1 $ and the continuum equations are obtained in the limit $\varepsilon\to 0$. We explain next the scaling assumptions considered.

\paragraph{(b) Scaling assumptions on the parameters.}
Three types of scaling assumptions are made: 
\begin{inparaenum}[(i)]
    \item the radius of alignment of the agents is supposed to be small and scaled as $r_A=\mathcal{O}(\sqrt\varepsilon)$;
    \item the agent-agent repulsion distance is supposed to be small and scales as $r_R = O(\varepsilon)$, but it is ensured that the potential stays of order 1 by setting
    \begin{equation} \label{eq:def_mu}
        \int \psi(x) dx = \mu <\infty;
    \end{equation}
    \item the agents alignment rate $\nu$ and orientational noise intensity $d_s$ in \eqref{eqn:IBMorient} are supposed to be very large and scale as: $d_s,\nu = \mathcal{O}(\frac{1}{\varepsilon})$ with $\frac{d_s}{\nu} = \mathcal{O}(1)$: this corresponds to fast agent-agent alignment and diffusion \cite{Degond2008}.
\end{inparaenum}

\paragraph{(c) Uniform anchor density and stiff regime assumptions.} It is  assumed that the anchor density for the obstacles is constant (uniformly distributed) and that the obstacles' springs are very stiff (the parameter $\kappa$ is very large). To this aim, we consider the ratio 
\begin{equation} \label{eq:def_gamma}
\gamma = \frac{\eta}{\kappa} \ll 1
\end{equation}
to be small. We suppose also a low obstacle noise regime\st{,} by considering the smallness of 
\begin{equation} \label{eq:delta}
\delta = d_o \gamma \ll 1 .
\end{equation}

\bigskip
The set of assumptions (a) is sufficient to derive continuum equations. 
The large-particle-limit or mean-field limit gives rise to kinetic equations for the obstacle density $f=f(t,x,y)$ and the agent density $g=g(t,z,\alpha)$. 
The set of assumptions (b) and (c) are sufficient to obtain closed equations for the obstacle density $\rho_f=\rho_f(t,x)$, the agent density $\rho_g=\rho_g(x,t)$ and the mean-agent orientation $\Omega=\Omega(x,t)$. In particular, the scaling assumptions $r_A=\mathcal{O}(\sqrt{\varepsilon})$ and $r_R=\mathcal{O}(\varepsilon)$ imply that alignment and agent-agent repulsion forces become localized in space as $\varepsilon\to 0$. The set of assumptions (c)  is used to Taylor expand the function $f$ with respect to $\gamma$ and $\delta$.

\bigskip
In summary, the continuum equations approximate a system with a very large number of agents and obstacles in the regime where the parameters of the system reach a given range of values, as described above, \textit{i.e.}, in the regime $\varepsilon\to 0$ (by an asymptotic analysis) and $\gamma \approx 0$, $\delta \approx 0$ (by a Taylor expansion approximation). 
These approximations {will be} taken into account when comparing discrete and continuum simulations, since they determine the range of validity of the continuum dynamics.

\subsubsection{The continuum model}
\label{sec:macro_equations}

The authors in \cite{Aceves2020} obtain the following equations for the dynamics of the density of agents $\rho_g(x,t)~\in~\R$ and their mean orientation $\Om(x,t)~\in~\mathbb{S}^1$ at a point $x~\in~\R^2$ at time $t~\geq~0$:
\begin{align}
&\pd_t \rho_g+\nabla \cdot \left(U\rho_g\right)=0, \label{eqmacro}\\
&\rho_g \pd_t \Om + \rho_g \left(V\cdot \nabla \right)\Om + d_3 P_{\Om^\perp} \nabla \rho_g=\g_s P_{\Omega^\perp} \Delta (\rho_g \Om), \nonumber
\end{align}
where
\begin{align*}
&U=d_1\Om -\frac{1}{\zeta}\nabla \bar\rho_f-\frac{\mu}{\zeta}\nabla \rho_g,\\
&V=d_2\Om -\frac{1}{\zeta}\nabla \bar\rho_f-\frac{\mu}{\zeta}\nabla \rho_g,
\end{align*}
where $\rho_f(x,t)$ is the obstacle density given by:
\begin{align} \label{eq:rhof}
&\rho_f/\rho_A=1+\frac{1}{\kappa}\Delta \bar\rho_g+\frac{1}{\kappa^2}\mathcal{N}(\bar\rho_g)-\frac{\eta}{\kappa^2}\pd_t\Delta \bar \rho_g+\mathcal{O}\left(\left( \frac{\eta}{\kappa}\right)^3\right),
&\mathcal{N}(\bar\rho_g):=\text{det}\mathbb{H}(\bar\rho_g),
\end{align}
where $\rho_A$ is the distribution of the anchor points in space (assumed to be constant and here taken to be equal to $1$ in the simulations and computations); $\mathbb{H}$ denotes the Hessian, `det' denotes the determinant, and we have defined 
\begin{equation} \label{eq:bar_symbol}
\bar\rho:=\rho * \phi,
\end{equation}
the convolution between $\rho$ and $\phi$, where $\phi$ is the repulsion kernel between agents and obstacles, Eq.~\eqref{potentials}. 
In the numerical simulations we {will} drop the higher order terms in $\eta/\kappa$ for $\rho_f$.
The model parameters are the friction constants $\zeta$, $\eta$, the obstacle-spring constant $\kappa$, and the agent-agent repulsion intensity $\mu$ given by Eq.~\eqref{eq:def_mu}.

\medskip

The friction coefficient $\g_s$ reads 
\begin{equation}\label{viscocoeff} 
\gamma_s = \frac{r_A^2}{8} \left(\frac{d_s}{\nu} + c_2\right).
\end{equation}

\medskip

The constants $d_1$, $d_2$ and $d_3$ are defined by
\begin{equation} \label{eq:dcoeff}
d_i=u_0 c_i,
\end{equation}
where $u_0$ is the agent speed, and $c_1$, $c_2$ and $c_3$ are explicit constants that depend only on the fraction $d_s/\nu$:
\begin{subequations}
\label{eq:c_constants}
\begin{align}
c_1 &= \int_0^{2\pi} \cos \theta \, m(\theta)\, d\theta,\\
c_2 &=\frac{\int_0^{\pi} \sin^2\theta\cos\theta \, m(\theta)h(\theta)\, d\theta}{\int_0^ {\pi}\sin^2\theta\, m(\theta) h(\theta)\, d\theta},\\
c_3 &=d_s/\nu,
\end{align}
\end{subequations}
where
$$
m(\theta) =\frac{1}{Z} \exp\left(\frac{\nu}{d_s} \cos \theta\right), \qquad Z:=\int_0^{2\pi}\exp(\nu \cos \theta/d_s) \, d\theta.
$$
and where the function $h$ does not have a explicit form but it is the solution to a differential equation. 
Specifically, $h(\theta)=g(\theta)/\sin(\theta)$ where
$g$ is the unique solution (for the exact functional space in which this unique solution is defined, the reader is referred to \cite[Lemma 2.3]{degonddimarco2015})
$$\frac{\nu}{d_s} \sin\theta \frac{\mbox{d}g}{\mbox{d}\theta}+ \frac{\mbox{d}^2g}{\mbox{d}\theta^2}= \sin\theta.
$$

 \bigskip
 For an explanation on the meaning of these equations the reader is referred to \cite{Aceves2020}. 
 We just point here that the system \eqref{eqmacro} for $(\rho_g,\Omega_g)$ corresponds to the so-called Self-Organised Hydrodynamics with Repulsion (SOHR) \cite{degonddimarco2015} in the case where $\nabla_x \bar \rho_f=0$ (\textit{i.e.} when there is no influence from the obstacles). 
 The SOHR is the continuum version of the Vicsek model with agent-agent repulsion \cite{degonddimarco2015}.

\begin{remark}[Approximation for $\rho_f$ and blow-up]
\label{rem:validity_rho}
The density $\rho_f$ may take negative values: in that case the continuum simulations {will be}stopped. Notice also that solutions may `blow-up' in the sense that particle densities may concentrate at points in space.
\end{remark}

\section{Patterns: phase diagrams}

\subsection{Discrete dynamics}\label{sec:micro_simus}

\subsubsection{Simulation set up}

We here show some simulations of the discrete model \eqref{eqn:IBMmodel_massless} to give an overview of the different types of patterns that emerge depending on the values of the parameters. 
Simulations are performed with $N=M=3000$ agents and obstacles initially distributed uniformly in the periodic domain $U = [0,1]\times [0,1]$. We also suppose that anchor points $Y_k$ for the obstacles are uniformly distributed in $U$, and fix the initial agent direction to $\pi/4$. 

\medskip
We consider the following expressions for the agent-agent and agent-obstacle {repulsion} potentials:
\begin{equation}\label{potentials}
\psi(x) =  \frac{6 \mu}{\pi r_R^2} \left(1 - \frac{|x|}{r_R}\right)_+^2,  \qquad \phi(x) = \frac{3 C_\phi}{2 \pi \tau} \left(1 - \frac{|x|}{\tau}\right)_+^2,
\end{equation} 
where 
$$
x^2_+=\left\{
\begin{array}{cc} x & \mbox{if }
x^2\in \R_+, \\
0 & \mbox{if }x<0.
\end{array}
\right.
$$
Therefore, both potentials are compactly supported and act in a radius $r_R>0$ for agent-agent repulsion and a radius $\tau>0$ for agent-obstacle repulsion. Notice that the constants have been chosen such that 
$$\mu=\int \psi(x) dx  \quad \mbox{ and }\quad C_\phi= \int |\nabla \phi|(x)dx .$$

\medskip
We fix a set of parameters as described in Table \ref{table_param}, and focus our study on the interplay between three parameters: the obstacle spring stiffness $\kappa$,  the agent friction $\zeta$ and the agent-agent repulsion intensity $\mu$. 

\medskip

\subsubsection{Phase diagram}
 Fig. \ref{simus_micro} shows the output of the simulations at time $t=10$: at this time agents and obstacles patterns seem to have reached a steady state. In this figure, agents' positions and their orientations are represented with black arrows and obstacle's positions with blue dots. 
 The output of the simulations are grouped into three panels: panel (A) corresponds to weak obstacle spring stiffness $\kappa = 10$, and  panels (B) and (C) correspond to mild $\kappa = 100$ and strong $\kappa = 1000$ obstacle spring stiffness, respectively. 
 Inside each panel, we arrange the simulations in a table: right-to-left columns correspond to increasing values of the agent-agent repulsion force $\mu$, bottom-to-top rows correspond to increasing values of the friction coefficient $\zeta$.
 Notice that the value for the agent-agent repulsion force $\mu$ is not taken the same in all panels. Indeed, the values for $\mu$  selected are the ones that make different patterns appear in the simulations. We will justify further the particular choice of the parameters after the linear stability analysis of the continuum equations. 
 Notice that the values for $\zeta$ are also different in panel (C). We refer the reader to the caption of Fig. \ref{simus_micro} for the exact choices  for the parameter values of $\mu$ and $\zeta$. Finally, we point out that the figures marked with a red cross are the ones for which the videos can be found in the supplementary material (see Appendix \ref{AppendixA} for more details).
  
   From Fig. \ref{simus_micro} we observe that a rich variety of agents' patterns emerges when varying the spring stiffness $\kappa$, the intensity of the agent-agent repulsion $\mu$, and the friction coefficient $\zeta$. 
 
We classify these patterns into 4 main types and we outline the parameter regions corresponding to each with frames of different colors in Fig. \ref{simus_micro}:
\begin{itemize}
\item Trails of agents (framed in red): agents organize into trails inside the obstacle pool. This behavior is mainly observed for weak and mild  obstacle spring stiffness ($\kappa=10$, panel (A) and $\kappa=100$, panel (B) of Fig. \ref{simus_micro}, respectively) 
\item Honeycomb organization of the agents (framed in orange): For small obstacle spring stiffness $\kappa=10$ (panel  (A) of Fig. \ref{simus_micro}) and mild agent-agent repulsion $\mu> 0.1$ (middle columns), we observe that the agents organize into fixed honeycomb structures, framing the obstacles which concentrate into aggregates of different sizes and shapes (not necessarily round). We point out that this pattern was not detected in the previous publication \cite{Aceves2020}.
\item Travelling bands of agents (framed in yellow): only observed for large values of the obstacle spring stiffness $\kappa=10^3$ and large agent friction with the environment $\zeta=5$, here the agents organize into bands perpendicular to their direction of motion. The width of the bands increases with the agent-agent repulsion intensity $\mu$ (from left to right plots of the first row of panel (C)). 
\item Clusters of agents (framed in green): agents organize into clusters more or less round depending on the regime of parameters. Cluster formation appears in all regimes of obstacle spring stiffness $\kappa=10,10^2,10^3$ (all three panels), and the size of the clusters changes depending on the obstacle spring stiffness $\kappa$ and on the agent-agent repulsion intensity $\mu$ but seems independent of the agent friction $\zeta$. Particularly, we observe that the cluster sizes increase with $\mu$, until a point is reached in which $\mu$ is so large that agent-agent repulsion counteracts all the other aggregation forces (right columns of Fig. \ref{simus_micro}). Moreover, the parameter $\mu$ acts as a phase transition parameter between different types of patterns. During the transition from clustered to near-homogeneous agent distributions with increasing repulsion intensity $\mu$, we observe a passage to other pattern types such as trails (for weak $\kappa=10$ or mild $\kappa=100$ obstacle stiffness), or honeycomb organizations (for weak obstacle stiffness).
Finally, we note that for large obstacle spring stiffness $\kappa$ and small agent friction $\zeta$ (bottom row of panel (C), simulations marked with a green star) we observe the formation of  'pinned' clusters where the agents are grouped into very small clusters that do not move (see Supplementary material and Appendix \ref{AppendixA} for access to the videos)
\end{itemize}
Each of these agent patterns is surrounded by obstacles that are kept at a given distance from the agents. This distance depends on the stiffness of the obstacles' springs $\kappa$ and the agent-agent repulsion intensity $\mu$. On one hand, if obstacles are loose enough (\textit{i.e.}, $\kappa$ is small), the repulsion force between the agents and the obstacles may be large enough to keep them both at approximately the obstacle-agent repulsion distance $\tau$ (defined in the potential $\phi$, Eq.  \eqref{potentials}). On the other hand, agent-agent repulsion opposes this effect, by giving the agent population force to go against the pressure exerted by the obstacle pool. We indeed observe that increasing the agent-agent repulsion force $\mu$ (left-to-right columns of Fig. \ref{simus_micro}) decreases the typical distance between the agent structures and the obstacles.

\begin{figure}
\centering
\includegraphics[scale=0.6]{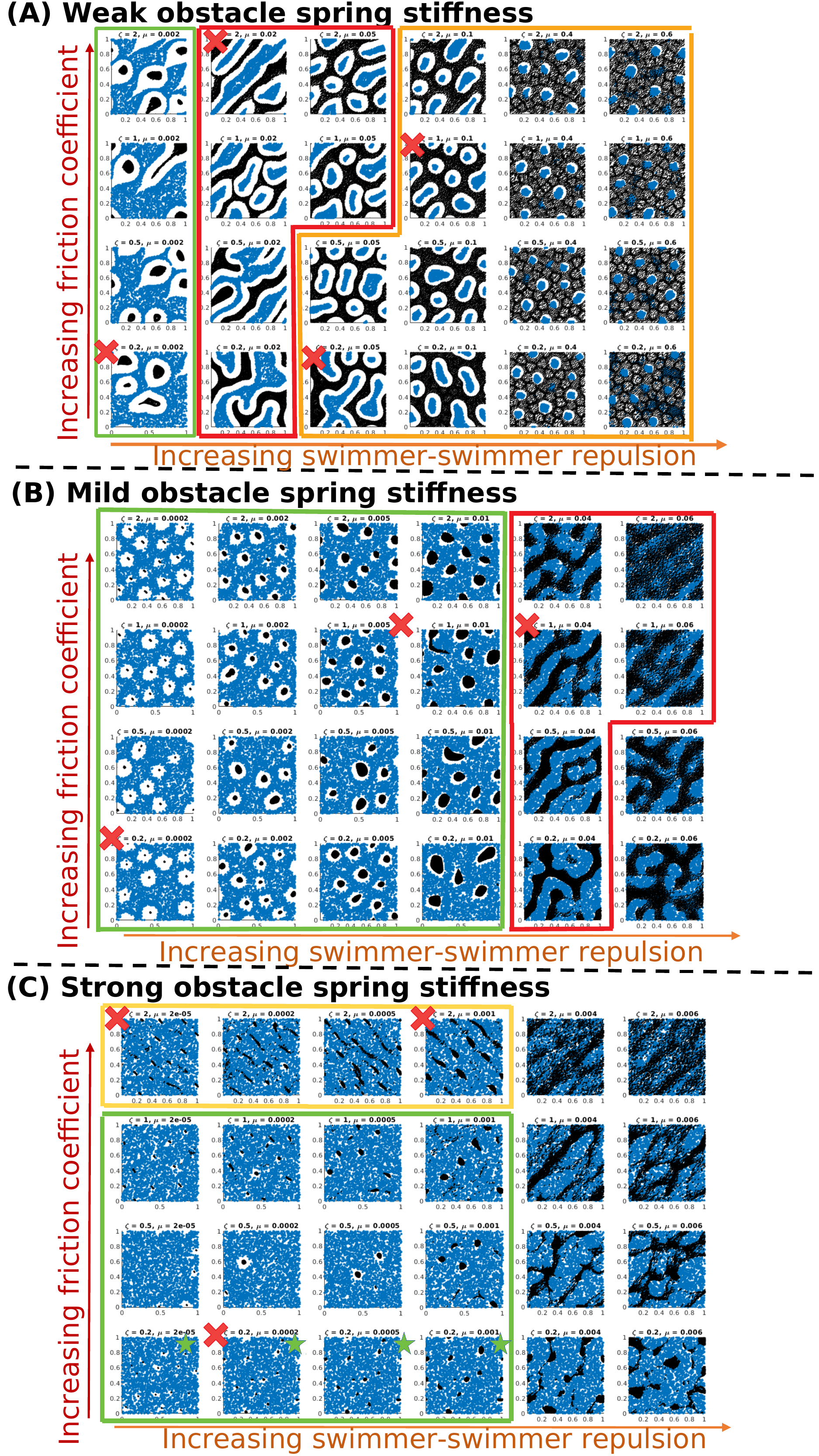}
\caption{Simulations of the discrete model for the parameters indicated in table \ref{table_param}. Agents are represented as black arrows giving their direction of motion, obstacles are represented as blue circles.  Panel  (A): for weak obstacle stiffness $\kappa = 10$, panel  (B): for mild obstacle stiffness $\kappa = 100$, panel  (C): for large obstacle stiffness $\kappa = 1000$. In each panel, the vertical axis represents different values of the friction coefficient $\zeta$ (from bottom to top: $\zeta = 0.2, 0.5, 1, 2$ for panels (A) and (B) and $\zeta=0.2,1,2,5$ for panel (C); and the horizontal axis represents different values of the agent-agent repulsion $\mu$: panel  (A): $\mu~\in~\{0.002, 0.02, 0.05, 0.1,  0.4, 0.6 \}$ , panel  (B):  $\mu~\in~\{0.0002, 0.002, 0.005, 0.01,  0.04, 0.06 \}$, panel  (C): $\mu~\in~\{2. 10^{-5}, 2. 10^{-4},6. 10^{-4},2. 10^{-3}, 4.10^{-3}, 6.10^{-3} \}$.  \label{simus_micro}}
\end{figure}

\begin{table}
\centering
\begin{tabular}{|c|c|c|}
\hline
Parameters & Value & Description\\
\hline
$N$ & 3000 & number of obstacles\\
$M$ & 3000 &  number of agents\\
$u_0$ & 1  & agent speed\\
$r_R$ & 0.075 & agent-agent repulsion distance\\
$r_A$ & 0.1 & agent-agent alignment distance\\
$\nu$ & 2 & agent-agent alignment intensity\\
$\tau$ & 0.15 & agent-obstacle repulsion distance\\
$C_\phi$ & 5 & agent-obstacle repulsion intensity\\
$d_s$ & 0.02 &  noise in the agents' orientation\\
$\eta$ & 1 & obstacle friction\\
$d_0$ & 0 & obstacle positional noise\\
$\mu$ & various & agent-agent repulsion intensity\\
$\zeta$ & various & friction constant of the agents\\
$\kappa$ & various & spring constant coefficient\\
\hline
\end{tabular}
\caption{Parameters used for the discrete simulations of Fig. \ref{simus_micro}. The various values considered for $\mu,\zeta,\kappa$ are specified in the caption of Fig. \ref{simus_micro}. \label{table_param}}
\end{table}

\subsection{Continuum dynamics}\label{MicMac}
In this section we show numerical simulations of the continuum equations \eqref{eqmacro} using the numerical scheme detailed in Appendix \ref{num_method}. 

\subsubsection{Simulation set up}\label{simus_mac}
We perform simulations of the continuum  model on the periodic domain $U = [0,1]\times [0,1]$ discretized with space step $\Delta x \approx 6.7 \; 10^{-3}$ ($150$ discretization points in each direction). The initial homogeneous agent direction $\Omega_0$ is set to $\frac{\pi}{4}$, and initial agent density $\rho_g$ is a small perturbation of a uniform distribution with $\int_\Omega \rho_g = 1$. In order to compare the numerical results with the discrete model, we use the same parameters as for the discrete simulations presented in Section \ref{sec:micro_simus} (see Table \ref{tab:table_param_macro}).

Notice that the agent-agent alignment distance at the continuum level is chosen to be $r_A = 0.15$ whereas for the discrete simulations it was $0.1$.
This choice corresponds to having rescaled $r_A$ approximately by a scaling factor $\varepsilon=0.5$, \textit{i.e.}, $r_A' = \sqrt{\eps}r_A$, where $r_A'=0.1$ is the parameter used in the discrete simulation (see the scaling assumption (b) in Sec. \ref{sec:assumptions}). 
Note that only the ratio $d_s/\nu$ is relevant for the continuous model, independently of their individual values. We therefore just ensure that this ratio is kept the same as for the discrete simulations and use $d_s/\nu=0.01$.
\medskip

\medskip

\begin{table}
\centering
\begin{tabular}{|c|c|c|}
\hline
Parameters & Value & Description\\
\hline
$h$ & $\approx 6.7\cdot 10^{-3}$  & step-size spatial discretization\\
$v_0$ & 1  & agent speed\\
$r_A$ & 0.15 & agent-agent alignment distance\\
$\frac{d_s}{\nu}$ & 0.01 & parameter coming from alignment forces\\
$\tau$ & 0.15 & agent-obstacle repulsion distance\\
$C_\phi$ & 5 & agent-obstacle repulsion intensity\\
$\eta$ & 1 & obstacle friction\\
$\gamma_s$ & $28\cdot 10^{-4}$ & viscosity coefficient \\
 $\mu$ & various & agent-agent repulsion intensity\\
  $\zeta$ & various & agent friction constant\\
  $\gamma$ & various & $\gamma = \eta/\kappa$\\
\hline
\end{tabular}
\caption{Parameters used for the simulations of the continuum equations \eqref{eqmacro} shown in  Fig. \ref{simus_macro}. The constants $d_1,d_2,d_3$ only depend on $\nu/d_s$ and are obtained by computing the expressions \eqref{eq:dcoeff} and \eqref{eq:c_constants}. \label{tab:table_param_macro}}
\end{table}

\bigskip
\subsubsection{Phase diagram}
\label{sec:macrovsmicro_simu}
We present the output of the continuum simulations. To facilitate the comparison with the discrete system, we adopt the same representation as the one presented in Fig. \ref{simus_micro}. 
In particular, the continuum densities are discretized as follows: at a simulation time $t$ we distribute randomly $N=3000$ agent points in the domain according to the distribution $\rho_g(\cdot,t)$, and similarly for the obstacle points using $\rho_f(\cdot, t)$. 
In Fig. \ref{simus_macro} we show the simulation results at the final time of the simulation, corresponding either to the time before blow-up or appearance of negative density for the obstacles (see Rem. \ref{rem:validity_rho}) or to $t=10$, as for the discrete simulations. 
As in Fig. \ref{simus_micro}, the simulations are separated in three panels: panel (A) is obtained for weak obstacle stiffness $\kappa = 10$, and panels (B) and (C) are for $\kappa = 100$ and $\kappa = 1000$ respectively. 
In each panel, we organize the simulations in tables for which bottom-to-top rows correspond to increasing values of the friction coefficient $\zeta$, while left-to-right columns correspond to increasing values of the agent-agent repulsion force intensity $\mu$.
See the legend of Fig. \ref{simus_macro} for more details on the parameter values considered for $\zeta$ and $\mu$.

\begin{figure}
\centering
\includegraphics[scale=0.6]{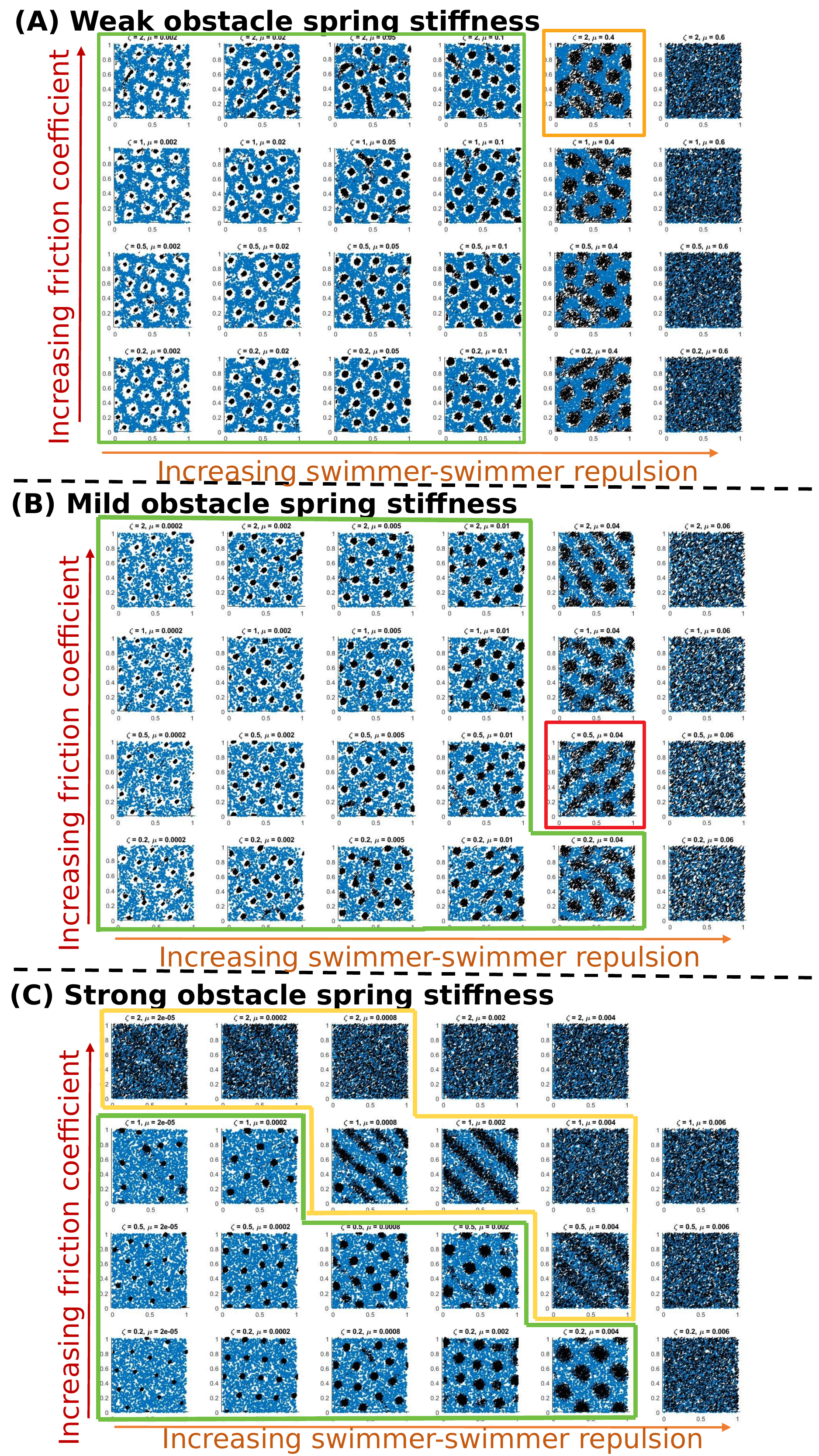}
\caption{Simulations of the continuum model \eqref{eqmacro} for the parameters indicated in table \ref{tab:table_param_macro}. Agents (randomly distributed from the distribution $\rho_g(x,t)$) are represented as black arrows of orientation $\pi/4$, obstacles (randomly distributed from the distribution $\rho_f(x,t)$) are represented as blue circles. panel  (A): for weak obstacle stiffness $\kappa = 10$, panel  (B): for mild obstacle stiffness $\kappa = 100$, panel  (C): for large obstacle stiffness $\kappa = 1000$. In each panel, the vertical axis represents different values of the friction coefficient $\zeta$ (from bottom to top: $\zeta = 0.2, 0.5, 1, 2$ and the horizontal axis represents different values of the agent-agent repulsion $\mu$: panel  (A): $\mu~\in~\{0.002, 0.02, 0.05, 0.1,  0.4, 0.6 \}$, panel  (B): left column: $\mu~\in~\{0.0002, 0.002, 0.005, 0.01,  0.04, 0.06 \}$, panel  (C): $\mu~\in~\{2. 10^{-5}, 2. 10^{-4},6. 10^{-4},2. 10^{-3}, 4.10^{-3}, 6.10^{-3} \}$.  \label{simus_macro}}
\end{figure}

In Fig. \ref{simus_macro} we observe different patterns for the agents, each framed using the same color code as for the discrete simulations: cluster formation (framed in green, present in all three panels), travelling bands (framed in yellow, panel (C), trails (framed in red, panel (B)), near-honeycomb structures (framed in orange, panel (A)), uniform distributions (unframed) and in-between states. Here again, increasing the obstacle spring stiffness $\kappa$ (from top to bottom panels) decreases the distance between agents and obstacles (\textit{i.e.}, the white area around the agents is reduced with increasing $\kappa$). We also observe that increasing the agent-agent repulsion intensity $\mu$ increases the size of the agent clusters and this parameter again serves as a transition parameter between clusters and uniform distribution of the agents, passing through honeycomb structures (first row of panel (A)), trails (third row of panel (B)) or travelling bands (first three rows of panel (C)). The effect of the friction parameter $\zeta$ becomes more relevant for large values of $\kappa$. For example, in panel (C) the parameter $\zeta$ serves as a transition parameter between clusters, trails and uniform states.

\paragraph{Comparing phase diagrams.}
We compare the two phase diagrams from the discrete simulations in Fig. \ref{simus_micro} and the continuum simulations in Fig. \ref{simus_macro}. Note, though, that there is not an exact correspondence of  the values for the parameter $\zeta$ used in panel (C) for the two cases.

It is noteworthy that the patterns observed with the continuum simulations are similar to the patterns of the discrete simulations (Fig. \ref{simus_micro}) for strong and mild obstacle spring stiffness (compare panels (B) and (C) of Figs. \ref{simus_macro} and \ref{simus_micro}), while the two models lead to different types of behavior in the weak obstacle stiffness regime (panel (A)). 
These are expected results since the continuum model has been obtained in a strong obstacle spring stiffness regime ($1/\kappa \approx 0$). As a result, the continuum model seems to be unable to produce the rich variety of patterns offered by the discrete model when considering loose obstacles. Also, we do not observe the pinned state with the continuum model, which appeared with the discrete dynamics when considering large obstacle spring stiffness $\kappa$ and small agent friction $\zeta$. 
Even though pinned-states are observed for large values of $\kappa$, they correspond to states where agents collapse into a very small cluster and then the numerical simulations of the continuum equations blow-up due to a high concentration of the agent density $\rho_g$ (see Rem.  \ref{rem:validity_rho}).


\section{Linear stability of uniform states}

\subsection{Analysis of the continuum model}
\label{sec:stability_analysis}

Continuum equations are amenable to linear stability analysis around constant solutions or uniform states. This is useful because the presence of instabilities signals the formation of patterns. In this section we obtain an explicit condition for the stability of uniform states.

\medskip
Before stating the main result, we introduce the following notation: denote by $\hat \phi$ the Fourier transform of $\phi$ defined as, for $k\in \R^2$:
\begin{align*}
\hat\phi_k:=\hat\phi_{|k|}=\hat \phi(k)=\int_{\R^2} e^{-ik\cdot x}\phi(x) dx \in \R.
\end{align*}
Notice that $\phi$ is assumed rotationally invariant, therefore  $\hat \phi$ is real and rotationally invariant (so we abused notation and wrote $\hat \phi_{|k|}$ instead of $\hat \phi_k$).

\begin{theorem}[Linear instability]
\label{th:linear_stability}
Consider fixed constant values
$\rho_0>0$ and $\Om_0\in \mathbb{S}^1$. Then, the linearized system  of \eqref{eqmacro} around $(\rho_0,\Omega_0)$ is unstable if and only if
\begin{equation} \label{eq:instability_condition}
\mbox{there exists $z> 0$ such that }z^2 (\hat \phi_z)^2>\mu\kappa.
\end{equation}
\end{theorem}

 The proof of the theorem is given later. 
 First, we derive sufficient conditions for the system to be stable:
\begin{corollary}[Conditions for stability]
\label{cor:stability}
Suppose that $\phi$ is absolutely continuous, rotationally invariant, and $\phi,\, \phi'\in L^1$. Then, it holds that
\begin{equation} \label{eq:c_0}
c_0:=\, \underset{z\in \R_+}{\max}z^2 (\hat \phi_z)^2<\infty
\end{equation}
and if $\mu\kappa >c_0,$
then the continuum equations \eqref{eqmacro} are linearly stable.

Moreover, if $\phi$ is given by \eqref{potentials}, define $c_0'=c_0/C_\phi^2$. It holds that the constant $c_0'$ is independent on the obstacle-agent repulsion radius $\tau$ and the intensity $C_\phi$ and the system is stable whenever
\begin{equation*} 
\frac{\mu\kappa}{C_\phi^2}>c_0'.
\end{equation*}
\end{corollary}
\begin{proof}
{Since by assumption $\phi$ is absolutely continuous and $\phi, \phi'\in L^1$, we have that $|\hat \phi'(k)|~=~|k||\hat \phi(k)|$. Moreover, since $\phi'\in L^1$, then $\hat \phi'$ is bounded. Therefore, $|k|^2|\hat \phi(k)|^2$ is bounded and} $c_0$ is finite. In this case, for $\mu\kappa>c_0$ the instability condition \eqref{eq:instability_condition} does not hold, so the system is stable.

In the particular case where $\phi$ takes the shape given in \eqref{potentials}, one can check that the following self-similarity condition holds
$$|k|\hat \phi_k = \tau |k| \hat \phi^{(1)}(\tau |k|),$$
where $\phi^{(1)}$ corresponds to $\phi$ when taking $\tau=1$. 
Therefore, it holds that
$$C_\phi^2 c_0=\max_k|k|^2(\hat \phi_k)^2 = \max_k(\tau |k|)^2 (\hat \phi^{(1)}(\tau |k|))^2=\max_y |y|^2 (\hat \phi^{(1)}(|y|))^2,$$
and so $c_0$ is independent of $\tau$. 
The rest of the corollary follows: the value of $c_0'$ is also clearly independent of $C_\phi$ as it is just a multiplicative factor of $\phi$.
\end{proof}

\begin{remark}[Limiting case of pillar obstacles]
 In the case where the obstacles are fixed pillars, \textit{i.e.}, the case where $\kappa\to \infty$, then the uniform distribution of agents and pillars is always a stable solution. The effect of this limiting case is that the equations for the agents on $(\rho_g, \Omega)$ become decoupled from the obstacles' density $\rho_f=\rho_A$, which is just constant (take the formal limit $\kappa\to \infty$ on the continuum equations \eqref{eqmacro}). Therefore, there is an abrupt behavioural change between static obstacles and obstacles that can move a bit (anchored at a fixed point \emph{via} a very stiff spring). This shows that, in this particular set up, the fact that the agents are able to modify their environment is {crucial} for interesting patterns to emerge.
\end{remark}

\paragraph{The role of the parameters.}
From the instability condition \eqref{eq:instability_condition}, we observe that the main drivers of the formation of instabilities are: the shape of the agent-obstacle repulsion potential $\phi$, the obstacle-spring stiffness $\kappa$, and the agent-agent repulsion intensity $\mu$.  
High agent-agent repulsion - high values of $\mu$ - has a stabilising effect while high agent-obstacle repulsion - high values of $C_\phi$ - has a destabilising effect, and vice\ls{-}versa. Also, high values of the spring constant $\kappa$ have a stabilising effect and small values have the opposite effect. 

From Cor. \ref{cor:stability}, in the case when $\phi$ is given in \eqref{potentials} the ratio given by 
\begin{equation} \label{eq:bifurcation_parameter}
b_p:=\frac{\mu\kappa}{C^2_\phi c'_0}
\end{equation}
 is the single value that acts as a bifurcation parameter.
However, the obstacle-agent repulsion radius $\tau$ plays a role in determining the size of the patterns (see Fig. \ref{simus_macro}). 
Also, from the instability condition \eqref{eq:instability_condition} and corollary \ref{cor:stability}, for typical shapes of the potential~$\phi$, we expect to have stability for small and large values of the {wave vector} $k$ but instabilities can appear at intermediate values whenever $c_0>\mu \kappa$ (where $c_0$ is given in \eqref{eq:c_0}).

\bigskip 
\noindent
The rest of this section is devoted to the proof of Th. \ref{th:linear_stability}. 
\begin{proof}[Proof of Th. \ref{th:linear_stability}]
We start by linearising the continuum equations \eqref{eqmacro} around $(\rho_0,\Omega_0)$ by expanding the solution using a small perturbation parameter $\beta$
\begin{equation} \label{eq:perturbation}
\rho_g = \rho_0 + \beta \rho_1+\mathcal{O}(\beta^2), \qquad \Om = \Om_0 + \beta \Om_1+\mathcal{O}(\beta^2), \qquad |\Omega|=1.
\end{equation}

\noindent Dropping the higher order terms, we obtain the linearised system (where the over-script bar notation is defined in Eq. \eqref{eq:bar_symbol}):
\begin{subequations} \label{eq:linSys}
\begin{align}
&\partial_t \rho_1 + d_1\Om_0\cdot \nabla \rho_1 + d_1\rho_0 \nabla\cdot \Omega_1=\bar\mu \rho_0 \Delta \rho_1+\rho_0\bar\lambda \left(\Delta^2 \bar{\bar \rho}_1-\gamma \Delta^2\pd_t \bar{\bar \rho}_1\right), \label{eq:linRho}\\
&\rho_0 \pd_t\Om_1 + \rho_0 d_2 \left(\Om_0\cdot \nabla\right)\Om_1+d_3 P_{\Om_0^\perp}\nabla \rho_1=\gamma_s\rho_0 P_{\Om_0^\perp}\Delta \Om_1,\label{eq:linOm}\\
&\Om_0\cdot \Om_1 =0, \label{eq:linNorm}
\end{align}
\end{subequations}
 where $\Delta^2$ is the bi-Laplacian, \textit{i.e.}, $\Delta^2 \rho= \Delta(\Delta \rho)$ and $P_{\Omega_0^\perp}$ is the orthogonal projection on $\Omega_0^\perp$. 
 Note also that $\bar\mu=\mu/\zeta$, $\gamma=\eta/\kappa$ and $\bar \lambda=\rho_A/(\kappa \zeta)$ (we assume $\rho_A=1$).

We now define the functions $F,G: \R_+\to \R$ by:
\begin{align}
\label{eq:FG}
&F(z):=z^2\frac{\rho_0}{\zeta}\left(\frac{1}{\kappa}z^2 (\hat \phi_z)^2-\mu\right),\\
&G(z):=1+\rho_0\frac{\eta}{\kappa^2 \zeta} z^4  (\hat\phi_z)^2>0,\nonumber
\end{align}
and given $k\in R^2$, we denote by  $k_0$, $k_1$ the quantities
\begin{equation}
\label{eq:abbrev}
k_0=(k\cdot\Omega_0), \quad k_1 = (k\cdot \Omega_0^\perp),
\end{equation}\newline
where $\Omega_0^\perp$ is the image of $\Omega_0$ by the rotation of angle $\pi/2$.
Th. \ref{th:linear_stability} is then a direct consequence of the following proposition. 
\begin{proposition}
\label{prop:linearStability}
System \eqref{eq:linSys} allows for non-trivial plane wave solutions, \textit{i.e.} solutions of the form
\begin{equation} \label{eq:wave-form-observables}
\rho_1(x,t)=\tilde\rho e^{i k\cdot x+\alpha t},\qquad \Om_1(x,t)=\tilde\Om e^{i k\cdot x+\alpha t},
\end{equation}
where $k\in\R^2$ is the wave vector, $\alpha\in\mathbb{C}$, $\tilde \rho\in\mathbb{C}$ and $\tilde\Om\in\mathbb{C}^2$, and $(\tilde \rho, \tilde \Omega)\neq (0,0)$ if and only if $\alpha$ and $k$ fulfil the following dispersion relations:\newline
\textbf{Case A:} $k \parallel \Omega_0$\newline
\textit{Option 1:} $\tilde \rho\neq 0$, $\tilde \Om=0,$
\begin{align}
\alpha= \alpha_1(k):=  -i\frac{d_1k_0}{G(|k_0|)} + \frac{F(|k_0|)}{G(|k_0|)}. \label{alphaA}
\end{align}
\textit{Option 2:} $\tilde \rho= 0$, $\tilde \Om\neq 0$.
\begin{align} \label{eq:stable_case}
\alpha= \alpha_2(k): = -id_2k_0 - |k|^2\gamma_s.
\end{align}
\textbf{Case B:} $k \nparallel \Omega_0$. Then, $\alpha$ is a root of the following polynomial of degree 2:
\begin{align}
\alpha^2 G &+ \alpha \left[G|k|^2 \gamma_s-F + ik_0 \left(Gd_2+d_1\right)\right] \label{alphaB}\\
&+d_1(\rho_0d_3k_1^2-d_2k_0^2)-|k|^2\gamma_s F + i\left(d_1k_0|k|^2\gamma_s-d_2k_0F\right)=0.\nonumber
\end{align}

The real parts of $\alpha$ are negative if and only if the following holds:
\begin{align}
\label{eq:RH1}
G(k)|k|^2 \gamma_s-F(k)>0,
\end{align}
and
\begin{align}
\label{eq:RH2}
H(k):=\left[G(k)|k|^2 \gamma_s-F(k)\right]^2&  d_1d_3 k_1^2\\- \gamma_sF(k)|k|^2&\left[(d_1-d_2G(k))^2k_0^2+\left[G(k)|k|^2 \gamma_s-F(k)\right]^2\right]>0.\nonumber
\end{align}

\end{proposition}

\begin{proof}[Proof of Prop. \ref{prop:linearStability}]
Substituting the plane wave ansatz into the equation yields
\begin{subequations} \label{eq:stab}
\begin{align}
&\tilde \rho \alpha + i \tilde \rho d_1 \left(\Om_0\cdot k\right) + i \rho_0 d_1 \left(\tilde \Om \cdot k\right)=-|k|^2\bar\mu\rho_0\tilde \rho + |k|^4 \bar\lambda \rho_0\tilde\rho(\hat\phi_k)^2\left(1-\gamma\alpha\right), \label{eq:stabRho}\\
&\rho_0\alpha\tilde \Om+i\rho_0d_2\tilde \Om \left(\Om_0\cdot k\right)+i\tilde\rho d_3 P_{\Om_0^\perp}k=-|k|^2\rho_0\gamma_s\tilde \Om,\label{eq:stabOm}\\
&\Om_0\cdot \tilde \Om =0. \label{eq:stabNorm}
\end{align}
\end{subequations}
or (if $\tilde \Omega = \omega \Omega^\perp_0$)
\begin{align*}
    & (G(|k|)\alpha - F(|k|)+i d_1 k_0) \tilde \rho + i \rho_0 d_1 k_1 \omega =0,\\
    &  i d_3 k_1 \tilde \rho + \rho_0 (\alpha + i d_2 k_0 + |k|^2 \gamma_s) \omega =0.
\end{align*}
This is a homogeneous linear system in $(\tilde \rho, \omega)$ which has a non-trivial solution if and only if the determinant of the system is $0$, \textit{i.e.}:
\begin{equation} \label{eq:condition_existence}
   \left( G(|k|) \alpha-F(|k|)+ i d_1 k_0 \right) \left(\alpha + i d_2 k_0 + |k|^2 \gamma_s \right) +  d_1 d_3 k_1^2=0.
\end{equation}
If $k_1=0$, there are two roots corresponding to either bracket being zero. This leads to \eqref{alphaA} or \eqref{eq:stable_case}. If $k_1\neq 0$, we can recast \eqref{eq:condition_existence} in \eqref{alphaB}.

To determine the sign of the real part of $\alpha$, we use the Routh-Hurwitz criterion for polynomials with complex coefficients \cite{gantmacher2005applications,morris1962routh}. In our case the Routh-Hurwitz criterion states that the $\text{Re}(\alpha)<0$ for all solutions $\alpha$ if and only if expressions \eqref{eq:RH1} and \eqref{eq:RH2} hold.

\end{proof}

With Prop. \ref{prop:linearStability} we conclude the proof of Th. \ref{th:linear_stability} as follows. 
Suppose \eqref{eq:instability_condition} holds and let $z_0>0$ be such that $z_0^2 \hat \phi^2_{z_0}>\mu \kappa$. Let $k= z_0 \Omega_0$. Then $k_0=z_0$ and $k_1=0$. So $F(|k|)=F(z_0)>0$ and $\alpha=\alpha_1(k)$ is such that $\mbox{Re }(\alpha)>0$. Hence, the linearized system is unstable.

Suppose now \eqref{eq:instability_condition} does not hold, \textit{i.e.}, $z^2 \hat \phi^2_z<\mu \kappa$, for all $z\in \R_+$. 
Then, $F(|k|)<0$, for all $k\in \R^2$. It results that Re$(\alpha_1(|k|))<0$, Re$(\alpha_2(|k|))<0$. Furthermore \eqref{eq:RH1} and \eqref{eq:RH2} are obviously satisfied for all $k\in \R^2$. Hence the system is stable.
\end{proof}


\subsection{Numerical validation of the linear stability analysis}\label{macstab}

In this section we compare the pattern predictions given by the linear stability analysis with the results obtained from numerical simulations.  This way we check that the linear stability analysis truly captures pattern formation, \textit{i.e.}, that nonlinear effects are of second order and most of the patterns characteristics are captured by linear effects.

\vspace{0.5cm}

\subsubsection{Predictions from the theoretical analysis and qualitative agreement with the macroscopic simulations}
We start by giving insights on the size and shape of the expected patterns based on the theoretical predictions offered by the stability analysis. To this aim, we consider perturbations introduced in the stability analysis (see Prop. \ref{prop:linearStability}), around the homogeneous density $\rho_0 = 1$ and in constant direction $\Omega_0 \in \mathbb{S}^1$ . As we are particularly interested in characterizing the patterns corresponding to clusters or bands, we will focus on the theoretical values for wave vectors parallel to $\Omega_0$ and parallel to $\Omega_0^\perp$ :
$$
k^{th}_\parallel = \underset{k \parallel \Omega_0 }{\mbox{argmax }} \mbox{Re}( \tilde \alpha(k) ), \quad k^{th}_\perp = \underset{k \parallel \Omega_0^\perp }{\mbox{argmax }} \mbox{Re}( \alpha(k) ), 
$$ 
where $\tilde \alpha(k)$ corresponds to case A  (Eq. \eqref{alphaA}), $\alpha(k)$ corresponds to case B (larger root of Eq. \eqref{alphaB}, computed numerically) and the symbol `Re' indicates the real part. With these wave vectors maximizing the real part of $\alpha$, we define the quantities:
$$
S^{th}_1 = \frac{{2\pi}}{|k^{th}_\parallel|}, \quad S^{th}_2 = \frac{{2\pi}}{|k^{th}_\perp|}.
$$
These quantities give the size of the expected patterns in each direction. We will also compute the maximal growth rates of the perturbations in these two directions:
$$
\alpha^{\parallel}_{max} = \underset{k \parallel \Omega_0 }{\mbox{max }} \mbox{Re}( \tilde \alpha(k) ), \quad \alpha^{\perp}_{max} = \underset{k \parallel \Omega_0^\perp }{\mbox{argmax }} \mbox{Re}( \alpha(k) ). 
$$
Equipped with these quantifiers, we now study the influence of the model parameters on the expected pattern shapes and sizes. As predicted by the stability analysis, patterns can be expected if the bifurcation parameter $b_p$ (Eq. \eqref{eq:bifurcation_parameter}) is below 1. We fix $C_\phi = 5$ and use $\phi$ as in Eq. \eqref{potentials} giving $c_0 \approx 5.6$ independent on $\tau$ as shown in the proof of corollary 1 (definition of $c_0$ in Eq. \ref{eq:c_0}). 
We vary $b_p$ by changing the values of the agent-agent repulsion intensity $\mu$ and aim to study the influence of the friction constant $\zeta$, the obstacle spring stiffness $\kappa$ and the agent-obstacle repulsion distance $\tau$. For each subsection, we compare qualitatively these predictions based on the linear stability analysis with simulations of the macroscopic model presented in Fig. \ref{simus_macro}

\vspace{0.5cm}

\paragraph{Influence of the friction constant $\zeta$.}
First we fix $\kappa = 1000$ and $\tau = 0.15$, and show in Fig. \ref{figstabzeta} the values of $\alpha^{\parallel}_{max}$ and $\alpha^{\perp}_{max}$ (left panel) and of $S^{th}_1$ and $S^{th}_2$ (right panel), as functions of the bifurcation parameter $b_p$ and for different values of the agent friction constant $\zeta$: $\zeta = 0.1$ (blue curves), $\zeta = 0.5$ (red curves), $\zeta = 1$ (yellow curves). 
\begin{figure}
 \includegraphics[scale = 0.35]{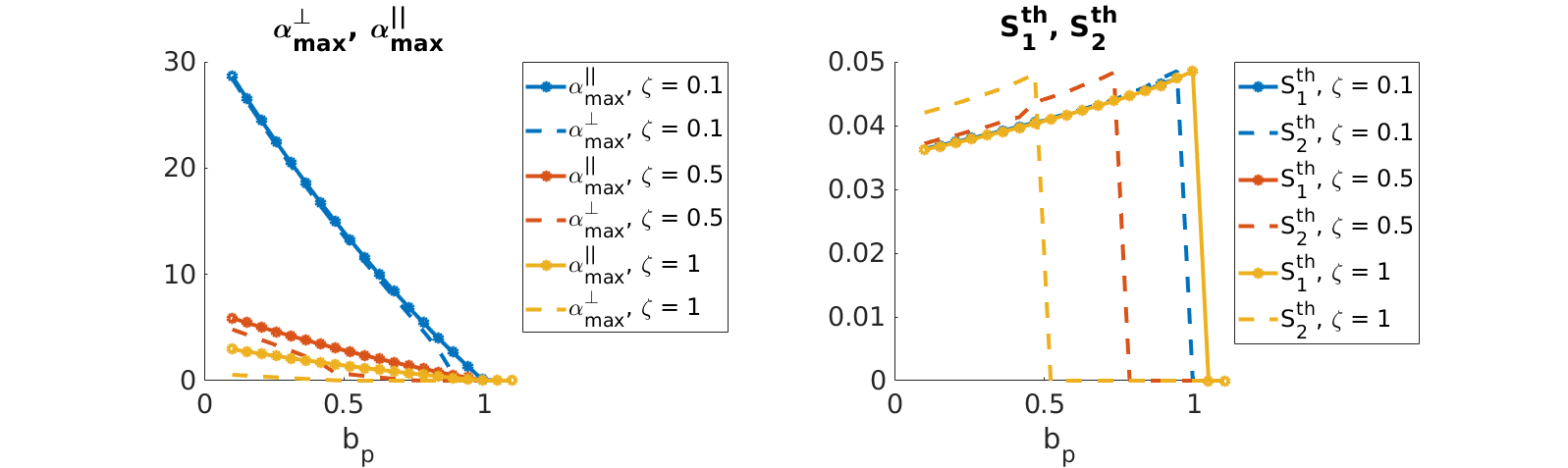}
\caption{Prediction of the linear stability analysis. Left: values of the maximal growth rate of the plane wave perturbations in the direction of $\Omega_0$ (continuous lines) and in the orthogonal direction $\Omega_0^{\perp}$ (dashed lines) as functions of the bifurcation parameter $b_p$, for different values of the agent friction $\zeta$: $\zeta = 0.1$ (blue curves), $\zeta = 0.5$ (red curves), $\zeta = 1$ (yellow curves). Right: same representation for the size of the perturbations in the two directions $S_1^{th}$ and $S_2^{th}$. }
    \label{figstabzeta}
\end{figure}
One can first observe in Fig. \ref{figstabzeta} (left panel) that we indeed recover the critical value 1 of the bifurcation parameter, below which perturbations grow ($Re(\alpha)>0$) and after which they are damped, independently on the value of $\zeta$. This shows that $b_p$ is indeed a relevant bifurcation parameter. Moreover, one can observe that perturbations grow faster for smaller values of the friction constant $\zeta$ (compare the blue and red curves in the left panel). From the right panel of Fig. \ref{figstabzeta}, we note first that the size of the clusters increases when increasing the bifurcation parameter (here, by increasing the agent-agent repulsion $\mu$). These are expected results as stronger agent repulsion leads to higher pressure in the agent population, leading to larger clusters. Secondly, we observe that the size of the patterns is independent on the friction constant $\zeta$, but the parameter zone in which patterns are of travelling band type (i.e $S_1^{th}>0$ and $S_2^{th} = 0$) is larger for larger values of $\zeta$ -compare the yellow and blue dashed curves in the right panel-. Thus, high friction substrates seem to favor the formation of travelling bands compared to low friction environments, provided the bifurcation parameter is large enough (large obstacle spring stiffness and/or large agent-agent repulsion compared to agent-obstacle repulsion).

\paragraph{Qualitative comparison with the macroscopic simulations.}
The influence of the agent friction $\zeta$ for $\tau=0.15$ and $\kappa = 1000$ can be observed in the macroscopic simulations presented in Fig. \ref{simus_macro} panel (C), comparing the rows together (from bottom to top for increasing values of $\zeta$). We first note that in the simulations of the three panels of Fig. \ref{simus_macro}, the values considered for the product $\mu\kappa$ were always the same, \textit{i.e.},
$$\mu\kappa \in \{0.02, 0.2, 0.5, 1, 4, 6 \},$$
and the value of $C_\phi = 5$ was kept constant, corresponding to the following values for the bifurcation parameter:
$$b_p \in \{0.0036,    0.0357,    0.0893,    0.1785,    0.7142,    1.0713\}.$$
We then observe that in each panel of Fig. \ref{simus_macro}, patterns are indeed observed in the first 5 columns of the tables while the last column displays a homogeneous distribution of agents. This validates the fact that patterns are observed only when the bifurcation parameter $b_p$ is below 1.

Moreover, focusing on the last panel (for which $\kappa = 1000$), we recover most of the observations predicted by the stability analysis: (a) the pattern size increases when increasing the bifurcation parameter (increasing $\mu$: compare simulations from left to right in Fig. \ref{simus_macro} panel (C)), (b) the zone of parameters showing travelling bands increases when increasing the agent friction $\zeta$ (compare bottom to top rows of panel (C)). Therefore, we obtain a very good qualitative agreement between the simulations of the macro model and the tendencies predicted by the stability analysis as function of $\zeta$.

\vspace{0.5cm}

\paragraph{Influence of the obstacle spring stiffness $\kappa$.}
Here we adopt the same representation as in the previous paragraph, but fixing the agent friction constant $\zeta = 0.5$ and playing on the obstacle spring stiffness $\kappa$ (we keep the agent-obstacle distance $\tau = 0.15$). Fig. \ref{figstabkappa} shows the values of $\alpha^{\parallel}_{max}$ and $\alpha^{\perp}_{max}$ (left panel) and of $S^{th}_1$ and $S^{th}_2$ (right panel), as functions of the bifurcation parameter $b_p$ and for $\kappa = 10$ (blue curves), $\kappa = 100$ (red curves) and $\kappa = 1000$ (yellow curves).
\begin{figure}
 \includegraphics[scale = 0.35]{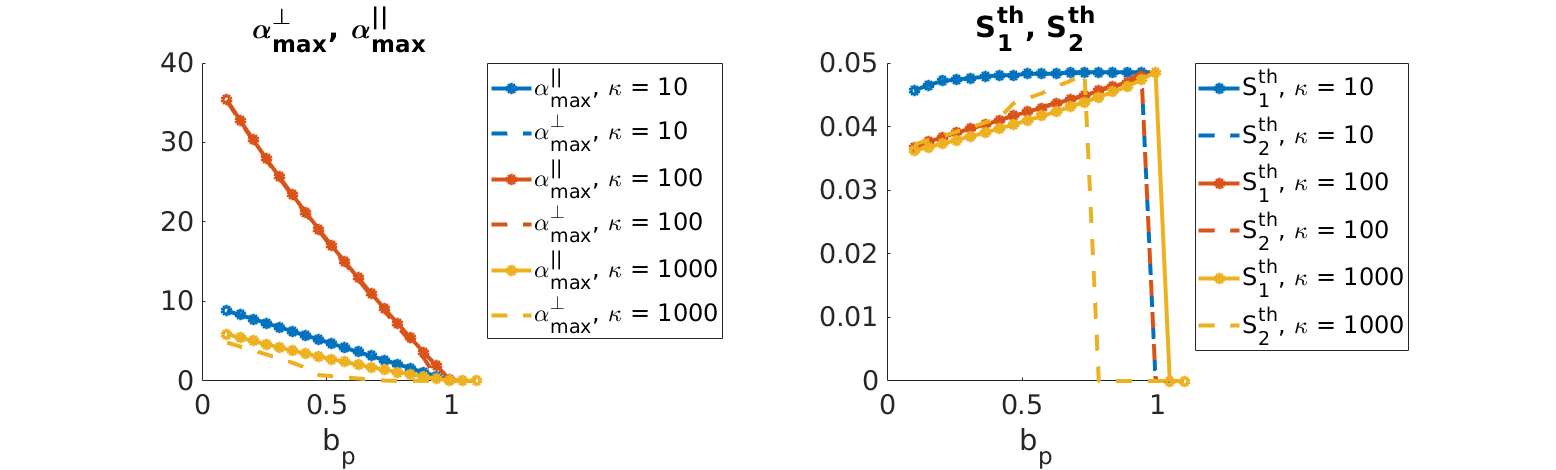}
\caption{Left: values of the maximal growth rate of the plane wave perturbations in the direction of $\Omega_0$ (continuous lines) and in the orthogonal direction $\Omega_0^{\perp}$ (dashed lines) as functions of the bifurcation parameter $b_p$, for different values of the obstacle spring stiffness $\kappa$: $\kappa = 10$ (blue curves), $\kappa = 100$ (red curves), $\kappa = 1000$ (yellow curves). Right: same representation for the size of the perturbations in the two directions $S_1^{th}$ and $S_2^{th}$. }
    \label{figstabkappa}
\end{figure}
From Fig. \ref{figstabkappa} (right), we can observe a similar evolution of the pattern size playing on the obstacle spring stiffness as when changing the friction constant $\zeta$: increasing the obstacle spring stiffness $\kappa$ slightly increases the zone of parameters favoring the formation of bands of agents (compare yellow and red dashed curves in the right panel). One can particularly note (blue curve of Fig. \ref{figstabkappa} (right)) that environments composed of loose obstacles ($\kappa = 10$) will only promote agent clusters the size of which is independent of the value of the bifurcation parameter. 
Finally we note from Fig. \ref{figstabkappa} (left) that the growth rate of perturbations does not evolve monotonically with the spring stiffness $\kappa$: faster perturbations are observed for $\kappa = 100$ compared to $\kappa = 10$ or $\kappa = 1000$ (compare red with blue and yellow curves in the left panel).

\paragraph{Qualitative comparison with the macroscopic simulations.}
The influence of the obstacle spring stiffness $\kappa$ for $\tau=0.15$ and $\zeta = 0.5$ can be observed in the macroscopic simulations presented in Fig. \ref{simus_macro}, comparing the second rows (starting from the bottom) in each panel (panel (A) for $\kappa = 10$, panel (B) for $\kappa = 100$, panel (C) for $\kappa = 1000$). 

Again, we obtain a very good agreement with the theoretical predictions: (a) the pattern sizes increase when increasing the bifurcation parameter (by increasing $\mu$: compare simulations from left to right in each panel), (b) the increase in pattern size as function of $\mu$ seems less important for $\kappa = 10$ (panel (A)) than for larger obstacle spring stiffness (panels (B) and (C)), and (c) travelling bands are only observed for $\kappa = 1000$ (panel (C)). 
\vspace{0.5cm}

\textbf{Influence of the agent-obstacle repulsion distance $\tau$}

Finally we aim to document the role of the agent-obstacle repulsion distance $\tau$. We adopt the same methodology as in the two previous paragraphs: we fix $\zeta = 0.5$ and $\kappa = 1000$ and show in Fig. \ref{figstabtau} the values of $\alpha^{\parallel}_{max}$ and $\alpha^{\perp}_{max}$ (left panel) and of $S^{th}_1$ and $S^{th}_2$ (right panel), as functions of the bifurcation parameter $b_p$ and for $\tau = 0.15$ (blue curves), $\tau = 0.2$ (red curves) and $\tau = 0.3$ (yellow curves). 
\begin{figure}
 \includegraphics[scale = 0.35]{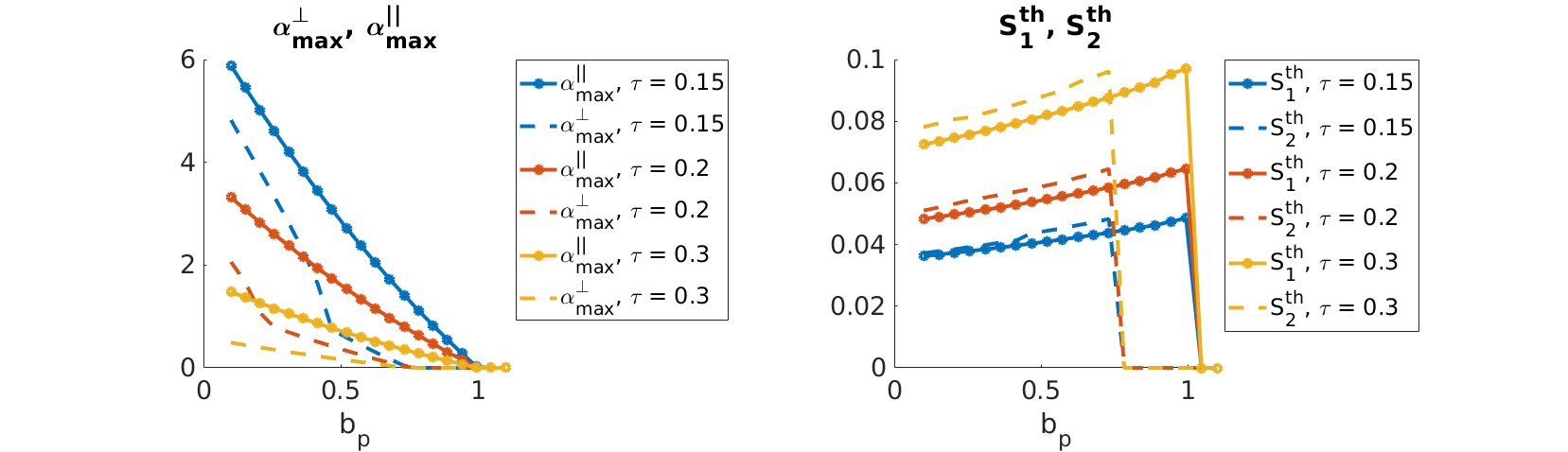}
\caption{Left: values of the maximal growth rate of the plane wave perturbations in the direction of $\Omega_0$ (continuous lines) and in the orthogonal direction $\Omega_0^{\perp}$ (dashed lines) as functions of the bifurcation parameter $b_p$, for different values of the agent-obstacle distance $\tau$: $\tau = 0.15$ (blue curves), $\tau = 0.2$ (red curves), $\tau = 0.3$ (yellow curves). Right: same representation for the size of the perturbations in the two directions $S_1^{th}$ and $S_2^{th}$. }
    \label{figstabtau}
\end{figure}
We first observe that increasing the value of $\tau$ slows down the growth of the perturbation modes (compare blue red and yellow curves of Fig. \ref{figstabtau} (left)). Moreover, as predicted by the stability analysis, the critical value of $\mu$ for which patterns appear does not depend on $\tau$: patterns are once again only observed as long as the bifurcation parameter $b_p$ does not exceed the value 1. Secondly, Fig. \ref{figstabtau} (right) shows that the agent-obstacle distance $\tau$ has a strong impact on the size of the clusters: larger $\tau$ leads to larger agent clusters (compare for instance yellow and blue curves in Fig. \ref{figstabtau} (right)), and agent-obstacle repulsion distance does not impact the shape of the patterns (clusters or bands types).  

As the simulations of Fig. \ref{simus_macro} have been generated only for $\tau = 0.15$, we are not able at this point to compare qualitatively the predictions of the stability analysis with the simulations of the macroscopic model as functions of this parameter. 
We will however assess the influence of $\tau$ via a quantitative comparison between the model and the theory in the next section. 

Altogether, these results show that agent-agent repulsion favors the spreading of the agents while agent-obstacle repulsion tend\ls{s} to aggregate the agents (and consequently cluster\ls{s} obstacles together). 
Travelling bands of agents seem to be favored in low friction environments composed of stiff obstacles, and the size of agent clusters seem to be controlled primarily by the agent-obstacle distance and the bifurcation parameter (ratio between the agent-agent repulsion intensity and the agent-obstacle repulsion intensity).

\subsubsection{Quantitative agreement between the macroscopic simulations and the stability analysis}
Here we provide a quantitative assessment of the pattern sizes computed numerically on the simulations of the macroscopic model and the ones predicted by the stability analysis. To this aim, we first compute numerically the pattern sizes using the 2D Discrete Fourier transform of the agent density at equilibrium $\hat{F} [\rho_g]=\hat{F} [\rho_g] (k)$, and extract the frequency of the two maximal modes $k_\parallel, k_\perp\in \R^2$ aligned in the direction of $\Omega_0$ and $\Omega_0^\perp$\ls{,} respectively:
$$
k_\parallel = \underset{k \parallel \Omega_0 }{\mbox{argmax}} \; |\hat{F} [\rho_g] (k)|, \quad k_\perp = \underset{k \parallel \Omega_0^\perp}{\mbox{argmax}} \; |\hat{F} [\rho_g] (k)|,
$$
where $ |\cdot| $ is the modulus of a complex number. 
Then, the theoretical quantifiers $S^{th}_1$ and $S^{th}_2$ will be compared with
$$S_1= \frac{2 \pi}{|k_\parallel|}, \quad S_2 = \frac{2 \pi}{|k_\perp|}.$$

In Fig. \ref{stability_macro}, we show the values of $S_1$ and $S_2$ (dotted curves) and $S_1^{th}, S_2^{th}$ (plain curves), for three different values of the obstacle spring stiffness $\kappa = 10$ (panel  (A)), $\kappa = 100$ (panel  (B)) and $\kappa = 1000$ (panel  (C)), and three different value of $\tau$: $\tau = 0.15$ (blue curves), $\tau=0.2$ (orange curves) and $\tau=0.3$ (yellow curves). Note that here $\zeta = 0.5$ so that theoretical predictions correspond to Fig. \ref{figstabtau}. Simulations are completed for $\Omega_0 = \left( \cos \frac{\pi}{4}, \;\sin \frac{\pi}{4} \right)$
and $\rho_0 = 1$.

As one can observe, we obtain a fairly good agreement between the values computed on the numerical solution and the ones predicted by the linear stability analysis as presented in Fig. \ref{stability_macro}. As predicted, the size of the repeating patterns increases as $\tau$ increases, (compare blue, red and yellow curves), and as the agent-agent repulsion intensity $\mu$ increases while staying below the critical threshold $\mu^*$ (above which the homogeneous steady-state profile is stable), corresponding to $b_p = 1$. 
For $\kappa=1000$ (panel  (C) of Fig. \ref{stability_macro}), we also recover the regime of travelling bands predicted for $\mu = 4 \cdot 10^{-3}$ and here $S_1 >0$ and $S_2=0$, \textit{i.e.}, patterns (the travelling bands) are only in the direction $\Omega_0$. 

\begin{figure}
\centering
\includegraphics[scale = 0.65]{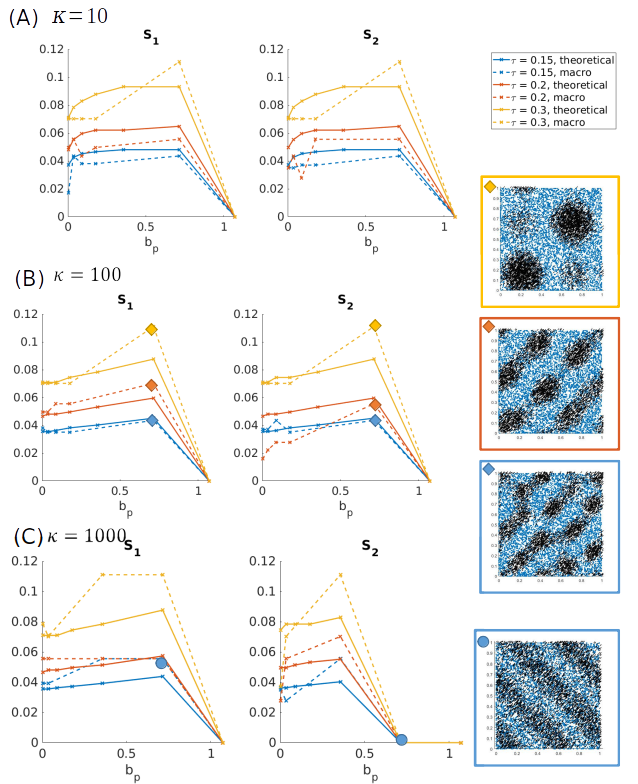}
\caption{Values of the maximal eigenmodes of the Fourier transform of the numerical solution (dotted curves) and predicted by the stability analysis (plain curves)  in direction $\Omega_0= (\cos\pi/4, \sin\pi/4)$ (left figures) and $\Omega_0^\perp$ (right coloured frames). 
Three different values of the obstacle spring stiffness are considered: $\kappa = 10$ (panel  (A)), $\kappa = 100$ (panel  (B)) and $\kappa = 1000$ (panel  (C)), and three different agent-obstacle repulsion force distances $\tau = 0.15$ (blue curves), $\tau=0.2$ (orange curves) and $\tau=0.3$ (yellow curves). Right column: examples of simulations for parameters reported on the graphs: simulations with diamond symbol match panel (B) (yellow frame: clusters; red and blue frame: trails) and the simulation with the circle symbol matches panel (C) (travelling bands). }
\label{stability_macro}
\end{figure}

\section{Quantitative assessment of the continuum model}\label{quantmicmac}

In this section, we aim to compare quantitatively the continuum  and discrete models.  As our goal is to compare continuous density profiles (continuum  model) with clouds of points representing individual positions (discrete model), a method to quantify the `proximity' between these two different types of solutions has to be devised. A first natural choice would be to use the quantifiers defined in  the previous section, \textit{i.e.} to compute the maximal eigenmode of the Fourier transform of the agent distributions from the discrete simulations. 
This would enable to construct a space-independent quantifier which could give an insight into the main structures of the discrete model.  
However, as one can observe in Fig. \ref{simus_micro}, the agent and obstacle structures that emerge from the discrete dynamics are not necessarily regularly spaced in the domain, which makes the use of the Fourier transform imprecise for the discrete simulations. 
In the following section, we propose a new method to compare discrete point clouds and continuum densities which does not require some spatial regularity of the patterns. 

\subsection{Methodology to compare discrete and continuum simulations}
\label{sec:methodology}

In Table \ref{tab:methodology}, we summarize the steps of the method we propose to compare discrete and continuum simulations. After generating two simulations (one with the continuum model and one with the discrete dynamics, step 1), we first aim to find the optimal Cartesian mesh on which (i) we interpolate the continuum solution and (ii) we compute the density of the point clouds using a Particle In Cell (PIC) method (step 2, Section \ref{sec:discretization_continuum}). 
At the end of this step, both solutions (continuum and discrete) are projected on the same Cartesian mesh.
In step 3 (section \ref{sec:compare_discretizations}), we then compute a Wasserstein-type distance based on the histograms of the two density distributions. 

\begin{table}[]
    \centering
    \begin{tabular}{p{0.2\textwidth} p{0.3\textwidth}p{0.05\textwidth}p{0.3\textwidth}}
        \textbf{STEP 1} & Simulation of the continuum equations \eqref{eqmacro} & & Simulation of the particle dynamics \eqref{eq:IBM_obstacle}-\eqref{eq:IBM_swim_positions} \\
        & \begin{center}$\Downarrow$ \end{center} & &  \begin{center}$\Downarrow$ \end{center} \\
         \textbf{STEP 2} & Discretization of the output density on a grid $\Pi^{\Delta x}_U$ using the PIC method: $\rho^{\Delta x}_{mac}$ & & Approximate the particle density on the grip $\Pi^{\Delta x}_U$ using the PIC method: $\rho^{\Delta x}_{mic}$\\
         &$\rightarrow$ \emph{Compute the optimal grid size $\Delta x$ using the $\ell^2$ distance (Sec. \ref{sec:discretization_continuum})} &&\\
         &\begin{center}$\searrow$ \end{center}&&\begin{center}$\swarrow$ \end{center} \\ 
       \textbf{STEP 3} &  \multicolumn{3}{c}{ Compare $\rho^{\Delta x}_{mic}$ and $\rho^{\Delta x}_{mac}$ with the distance $W(\rho^{\Delta x}_{mic}, \rho^{\Delta x}_{mac})$}\\
       &  \multicolumn{3}{c}{$\rightarrow$ \emph{Computed with the EMD method} (Sec. \ref{sec:compare_discretizations})}
    \end{tabular}
    \caption{Diagram of the methodology used to compare the simulations for the continuum equations \eqref{eqmacro} and the simulation of the discrete dynamics \eqref{eq:IBM_obstacle}-\eqref{eq:IBM_swim_positions}. PIC: Particle-In-Cell; EMD: Earth Movers Distance; $W$: Wasserstein distance.}
    \label{tab:methodology}
\end{table}

\subsubsection{Discretization of the particle density}
\label{sec:discretization_continuum}

A natural choice for comparing point clouds and continuum densities is to choose a Cartesian grid for both models, and compute the density of the individual agents using for instance a Particle-In-Cell (PIC) method \cite{degond2015}. However, the choice of grid points spacing is critical, as it depends on the profile of the distribution as well as on the number of particles present in the computational domain:  highly clustered agent distributions require fine meshes to enable to capture the characteristics of the small and concentrated agent clusters, while more homogeneous agent distributions require coarser grids to allow the capture of larger patterns (see Fig. \ref{procedure}). 
To be efficient, the grid spacing must, therefore, account for the characteristic size of the continuum  structures that can be captured with a finite number of individual points. As we want to compare a continuum model with a discrete one, we will use the continuum  simulations as a reference. 
Our goal here is to find the optimal Cartesian mesh on which a continuum density $\rho^{\Delta x}(x,t)$ would be best represented by a cloud of $N$ points, for $N$ given.
Note that the continuum density $\rho^{\Delta x}(x,t)$ is itself already discretized on a Cartesian mesh with spacing $\Delta x$, because it corresponds to a solution of the discretized continuum model.

Given a continuum density profile $\rho^{\Delta x}(x,t)$  - discretized on a Cartesian mesh $\Omega_{\Delta x} \subset \Omega$ with grid spacing $\Delta x= \frac{1}{N_x}$ in each direction - we first throw $N$ individual points $(y_1, \hdots, y_N) \in \Omega$ according to the distribution $\rho^{\Delta x}(x,t)$. 
We now denote by $\rho_{PIC}^{h}(y_1,\hdots, y_N)$ the density of the individual points $(y_1, \hdots, y_N)$ computed on a Cartesian mesh of spacing $h>0$ using a PIC method, and $\Pi^{\Delta x} \big(\rho_{PIC}^h \big)$ its linear interpolation on the initial mesh $\Omega_{\Delta x}$. 
We aim at finding the optimal grid spacing $h$ minimizing the $L^2$ distance between the initial continuum density $\rho^{\Delta x}(x,t)$ and its approximation by $N$ individual points:
$$
\tilde{h} = \underset{h}{\mbox{argmin}}\quad  || \rho^{\Delta x} - \Pi^{\Delta x} \big(\rho_{PIC}^h(y_1,\hdots,y_N) \big)||_{\ell^2(\Omega_{\Delta x})},
$$
\noindent where $||.||_{\ell^2(\Omega_{h})}$ denotes the discrete $l^2$ norm on a Cartesian mesh $\Omega_{h}$: $$||\rho^{h}||_{\ell^2(\Omega_h)} = h^2\sum_{i=1}^{N_x}\sum_{j=1}^{N_x} |\rho^{h}(x_i,y_j)|^2$$

The optimal $\tilde{h}$ is computed numerically. It therefore corresponds to the best grid spacing one can hope for approximating a density $\rho^{\Delta x}$ with a set of $N$ points. We therefore will use this quantity to compare a simulation of the continuum model with one of the discrete model performed with $N$ agents. The discretized macroscopic density will be denoted by $\rho^{\tilde h}_{mac} = \rho_{PIC}^{\tilde{h}}(y_1,\hdots,y_N)$, and the approximation from the discrete particle simulation will be denoted by $\rho^{\tilde{h}}_{mic}$ (computed via the PIC method on a grid with spacing $\tilde{h}$). In the following section, we describe how to compare $\rho^{\tilde{h}}_{mac}$ with $\rho^{\tilde{h}}_{mic}$.

\begin{figure}
\centering
\includegraphics[scale = 0.7]{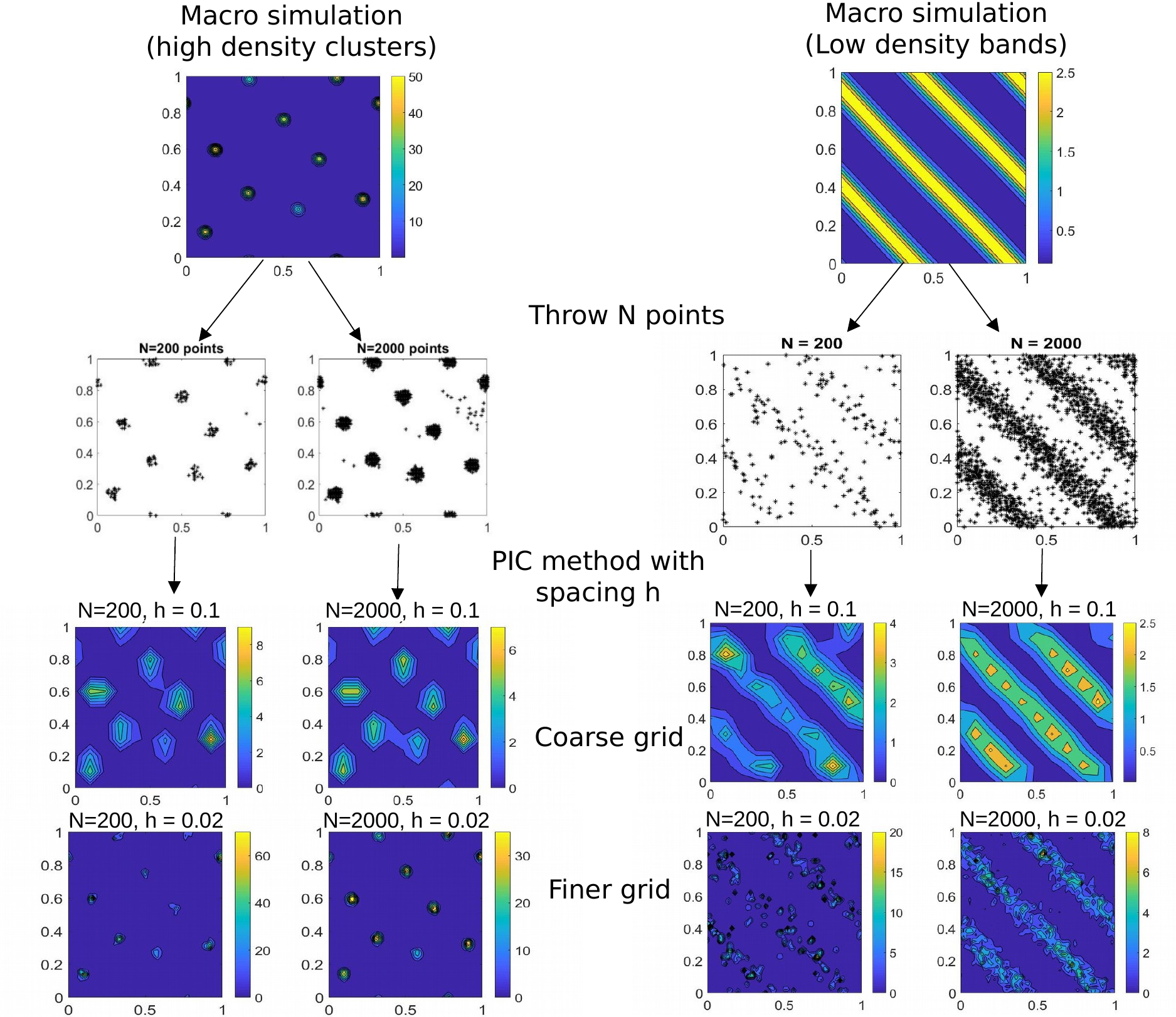}
\caption{Two examples of the procedure for choosing the optimal grid for the PIC method of the discrete simulations, starting from a continuum simulation with high density aggregates (left column) or low density bands (right column). The first step (first to second rows) consists in distributing $N$ points according to the continuum density distributions (left: for $N=200$, right for $N=2000$ points), and the second step (third and fourth rows) computes the approximated density using a PIC method with different spacing from the point distributions (third row: using a coarse grid with spacing $h = 0.1$, fourth row: using a finer grid $h = 0.02$). 
As one can see in Fig. \ref{procedure}, while the number of points to throw to approximate the continuum density does not play a major role for high density clustered distributions, it becomes critical for approximating more homogeneous distributions (compare left and right columns). Moreover, high density clusters require the use of a fine enough grid to correctly recover the initial distribution (compare third and fourth rows in the left column), while smoother distributions are better approximated using a large number of agents and coarse grids (third row of the left columns). These first results highlight the necessity for adapting the numerical grid used to compute the density of agents from the discrete model if one hopes to have a consistent quantifier to compare with the continuum model.  \label{procedure}}
\end{figure}

\subsubsection{Comparing discretized and discrete dynamics}
\label{sec:compare_discretizations}
The comparison between the discretized and the discrete dynamics will be done in several steps:
\begin{enumerate}[label=Step \arabic*),leftmargin=1.5cm]
    \item \label{wass} \textit{Choosing the right distance to compare the micro- and macro- simulations:} 
    we want to construct a quantifier enabling to compute the distance between the two distributions $\rho^{\tilde h}_{mac}$ and $\rho^{\tilde h}_{mic}$ described in the previous section (solutions of the continuum and discrete models projected on a Cartesian mesh with spacing $\tilde h$). The first natural choice would be to use the discrete $L^2$ norm as both quantities are defined on the same meshes. However, we need a quantifier independent on space translations, as there is no reason for the patterns of the discrete model to match exactly the locations of those of the continuum model at a given time. For example if the discrete and continuum simulations produce band patterns with same width and speed but not at the same positions, we still want to consider that the two solutions are very close to each other. Therefore, we propose here to use a Wasserstein-like distance. \\
    Inspired from \cite{rubner2000}, we choose to work with the Earth Movers Distance (EMD). The EMD is based on the minimal cost that must be paid to transform one distribution into the other  and relies on the solution to a transportation problem issued from linear optimization. As solving the transport problem in 2 dimensions is very costly, we 'compress'/approximate the density distributions using their signatures (histograms).
    \item \textit{Construction of the signatures of the distributions:} given a density profile on a grid containing $N_h = \frac{1}{h}$ points in each direction $(\rho_{ij})$, $i=1\hdots N_h, j=1 \hdots N_h$, the signature of $\rho$,  $P[\rho]=\{(p_1,\omega_{1}), \hdots, (p_m,\omega_{m}))\}$ is defined as:
    \begin{equation}
        p_k = \frac{k M }{n_b}, \quad \omega_{k} = \sum_{i=1}^{N_h} \sum_{j=1}^{N_h} \mathds{1} _{[p_{k-1}, p_{k}]} \big(\rho_{ij}\big), \quad k=1 \hdots n_b,
        \label{signatures}
    \end{equation}
    \noindent where $M = ||\rho||_\infty = \max_{i,j} \rho_{ij}$ and  the number of bins $n_b$ has been chosen using the Freeman Diaconis rule, for which the bin width corresponds to $2 \frac{IQR}{n^{3/2}}$, where $IQR$ is the interquartile range of the data and $n$ is the number of observations (in our case the number of grid points, $n = \frac{1}{h^2}$). We give in Fig \ref{signature} a visual representation of computing the signature of a toy distribution with 4 bins and in Fig. \ref{histo} an example of the histograms of two simulations of the continuum model. Note that when computed on density distributions, the points $p_k$ in each cluster correspond to local density values and the corresponding weights $\omega_k$ are the number of grid (spatial) points in which the density is comprised between the values $p_{k-1}$ and $p_{k}$. 
    \item \textit{Definition of the EMD between two signatures:} following the lines of \cite{rubner2000}, we apply the following linear programming problem: Let $P=\{(p_1,\omega_{1}), \hdots, (p_m,\omega_{m}))\}$ and   $Q=\{(q_1,v_{1}), \hdots, (q_n,v_{n}))\}$ be two signatures with $m$ and $n$ clusters represented by their representatives $p_k,q_{\ell}$ and their respective weights $\omega_k,v_{\ell}$ for $k = 1 \hdots m, \ell=1 \hdots n$. We want to find a flow $F=(f_{k\ell})$ minimizing the overall cost:
    $$
    W(P,Q,F) = \sum_{k=1}^m\sum_{\ell=1}^n d_{k \ell}f_{k \ell},
    $$
    where $d_{k \ell}$ is the ground distance matrix between clusters $p_k$ and $q_\ell$:
    $$
    d_{k \ell} = |p_k - q_\ell|.
    $$
    The minimization is made under the following set of constraints:
    \begin{align}
        &f_{ij} \geq 0, \quad 1\leq i \leq m, \; 1\leq j\leq n, \label{C1}\\
        &\sum_{j=1}^n f_{ij} \leq \omega_{pi}, \quad 1\leq i \leq m \label{C2}\\
        &\sum_{i=1}^m f_{ij} \leq \omega_{qj}, \quad 1\leq j \leq n \label{C3}\\
        &  \sum_{i=1}^m\sum_{j=1}^n f_{ij} = min(\sum_{i=1}^m \omega_{pi}, \sum_{i=1}^m \omega_{qj}).\label{C4}
    \end{align}
    \begin{figure}
        \centering
        \includegraphics[scale = 0.6]{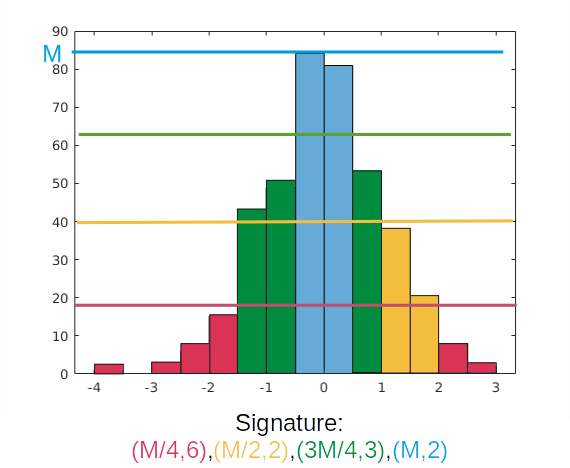}
        \caption{Example of the signature of a distribution with 4 bins, $M=85$ and $N_x = 13$. The different colors represent the different compartments of the signature. \label{signature}}
    \end{figure}

    \noindent If we look at the signatures $P$ and $Q$ as a set of goods at given locations (represented by $p$ and $q$) each with a given amount (represented by the weights $\omega$ and $v$), the EMD can be seen as a transportation problem consisting in finding the least expensive flow of goods from the suppliers to the consumers, where the cost of transporting a single unit of goods is given. Then, constraint \eqref{C1} expresses that 'supplies' can be transported from $P$ to $Q$ only, while constraints \eqref{C2}, \eqref{C3} limits the amounts of supplies that can be given by $P$ to $Q$ and that can be received from $Q$ to $P$\ls{,} respectively. 
    The final constraint \eqref{C4} expresses the fact that the total amount of mass transported must be optimal. Once this transportation problem is solved, the EMD between signatures $P$ and $Q$, $EMD(P,Q)$ is then defined as:
    $$
    EMD(P,Q) = \frac{\sum_{i=1}^m\sum_{j=1}^n d_{ij}f_{ij}}{\sum_{i=1}^m\sum_{j=1}^n f_{ij}}.
    $$

    \begin{figure}
        \centering
        \includegraphics[scale = 0.7]{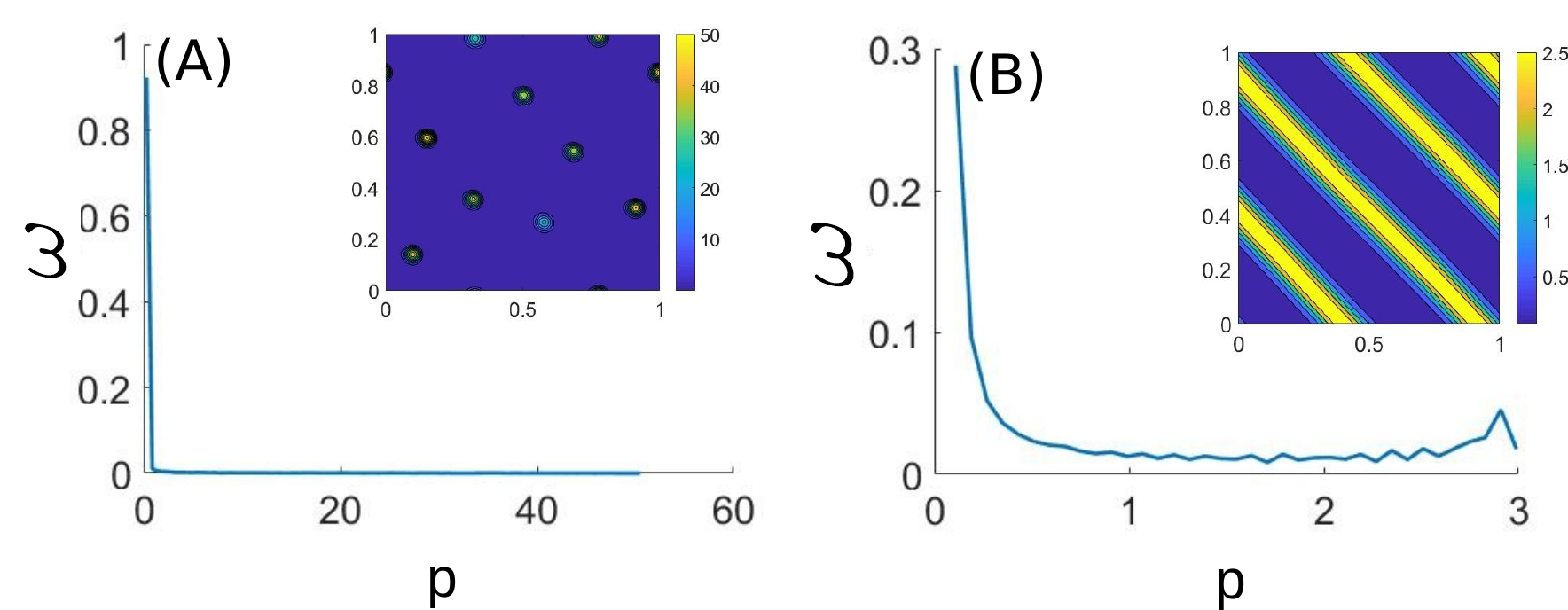}
        \caption{Histograms for simulations of Fig. \ref{procedure} as defined in Eq. \eqref{signatures}. \label{histo}}
    \end{figure}

    Rubner et al. proved in \cite{rubner2000} that when the ground distance is a metric and the total weights of the two signatures are equal, the EMD is a true metric. Therefore, by considering the Euclidean distance as ground distance we can use the EMD as a valid dissimilarity measure between signatures. However, as two different density distributions may have the same signature, the EMD with \eqref{signatures} as signatures is a pseudo-metric. However, as shown in Fig. \ref{histo}, the histograms between band like patterns and clustered state are very different distributions, making this pseudo metric suitable for measuring the dissimilarity between pattern types. Moreover, we check carefully in the next paragraph the validity of the EMD when it can be compared to the classical $L^2$ distance.
    \item \textit{Validation of the pseudometric EMD.}
    In order to check the validity of the pseudometric constructed in this section, we  aim to compare the efficiency of the EMD in cases where it can be compared to the classical $L^2$ distance. More specifically, we use it to measure the dissimilarity between the density profile of the continuum  simulation (high density clustered simulation of Fig. \ref{procedure} left column) and its approximation by a cloud of $N$ points reconstructed on a grid, using the procedure described in the previous section. 
    We show these dissimilarity measures in Fig. \ref{error}, as a function of the number of grid points for the PIC method ($N_{PIC}$, horizontal axis of Figs. \ref{error}) and different numbers $N$ of discrete particle (different curves), using the EMD distance based on histograms (left panel) or the $L^2$ distance based on point values (right panel). 
    As one can observe in Fig. \ref{error}, the EMD and the $L^2$ norm are in good accordance. 
    As previously observed in Sec. \ref{sec:discretization_continuum}, both metrics show that for each number of particles used to approximate the continuum  density distribution, there exists an optimal number of grid points for the PIC method which minimizes the distance between the initial density and its approximation by particles. 
    As expected, this optimal value increases as the number of particles increases, suggesting that using a larger number of agents allows the use of finer grids which enables us to better capture the fine structures of the continuum density distribution. 
    Moreover, this figure shows that the Wasserstein distance based on the EMD between density signatures seems to be a valid tool to compare density distributions. 

    \begin{figure}
        \centering
        \includegraphics[scale = 0.6]{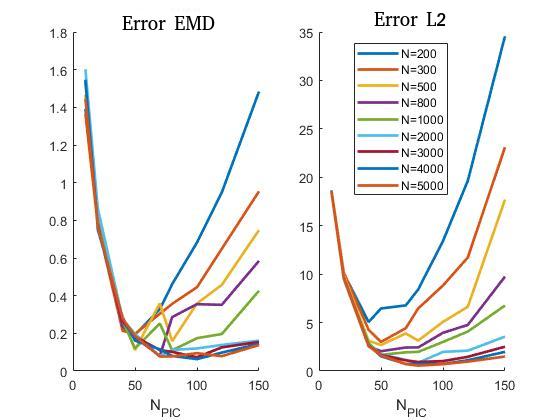}
        \caption{Error between the density profile of the continuum  simulation (high density clustered simulation of Fig. \ref{procedure} left column) and its approximations using the procedure described in Sec. \ref{sec:discretization_continuum}, as a function of the number of grid points for the PIC method $N_{PIC}$ (horizontal axis) and different number $N$ of discrete particles (see insert for correspondance between curve color and $N$), using the EMD distance based on histograms (left {panel}) or the L2 distance based on point values (right {panel}) \label{error}}
    \end{figure}

\end{enumerate}

In the next section we present  the numerical comparison between the discrete and continuum  models. 

\vspace{1cm}

\subsection{Results} \label{sec:results}

We aim to compare quantitatively the steady-states of the discrete and continuum  models in different regimes of the parameters, and study the influence of the number of agents for the discrete model $N$ as well as the scaling parameter $\epsilon$. 
We recall that the assumptions for the derivation of the continuum equations are given in Sec. \ref{sec:assumptions}. 
In particular, some of the parameters are scaled by a factor $\eps \ll 1$ in the following way (denoting by a tilde the values used for discrete simulations):
\begin{equation} \label{eq:rescaling_parameters}
\tilde{r_R} = \epsilon r_R, \quad \tilde{r_A} = \sqrt{\epsilon} r_A, \quad \tilde{d_s} = \frac{d_s}{\epsilon}, \quad \tilde{\nu} = \frac{\nu}{\epsilon}.
\end{equation}
For all  simulations, we consider the same number of agents and obstacles and set $M=N$, and we fix the values of $C_\phi = 5$ (leading to $c_0 = 5.6$) and $\zeta = 0.5$. 
For each set of parameters, we use the method previously described in Sec. \ref{sec:methodology} to compare discrete and continuum simulations.

\subsubsection{Mild obstacle spring stiffness}

In Fig. \ref{micmack100}, we show the simulations obtained for mild obstacle spring stiffness $\kappa = 100$. 
The left panel is obtained for $\mu = 2.10^{-3}$ (corresponding to a bifurcation parameter $b_p \approx 0.036$), and the right panel is for $\mu = 4.10^{-2}$ (corresponding to $b_p \approx 0.7$, close to the stability threshold 1). Top figures show the EMD between the continuum  and discrete solutions as a function of $\epsilon$, for different number of agents used for the discrete simulations $N$: $N=500$ (blue curve), $N=1000$ (red curve), $N=3000$ (yellow curve) and $N=5000$ (purple curves). The corresponding simulations are shown below in tables: for each, the left column shows the simulations of the continuum model, and the next columns are simulations of the discrete model for different values of $\epsilon$: $\epsilon = 0.1$ (second column), $\epsilon = 0.5$ (third column), $\epsilon = 0.8$ (fourth column), $\epsilon = 1$ (last column). The different rows of the tables correspond to different number of agents for the discrete simulations as well as for the discretization of the continuum density (from top to bottom: $N=500$, $N=1000$, $N=3000$, $N=5000$). 

\begin{figure}
\includegraphics[scale = 0.52]{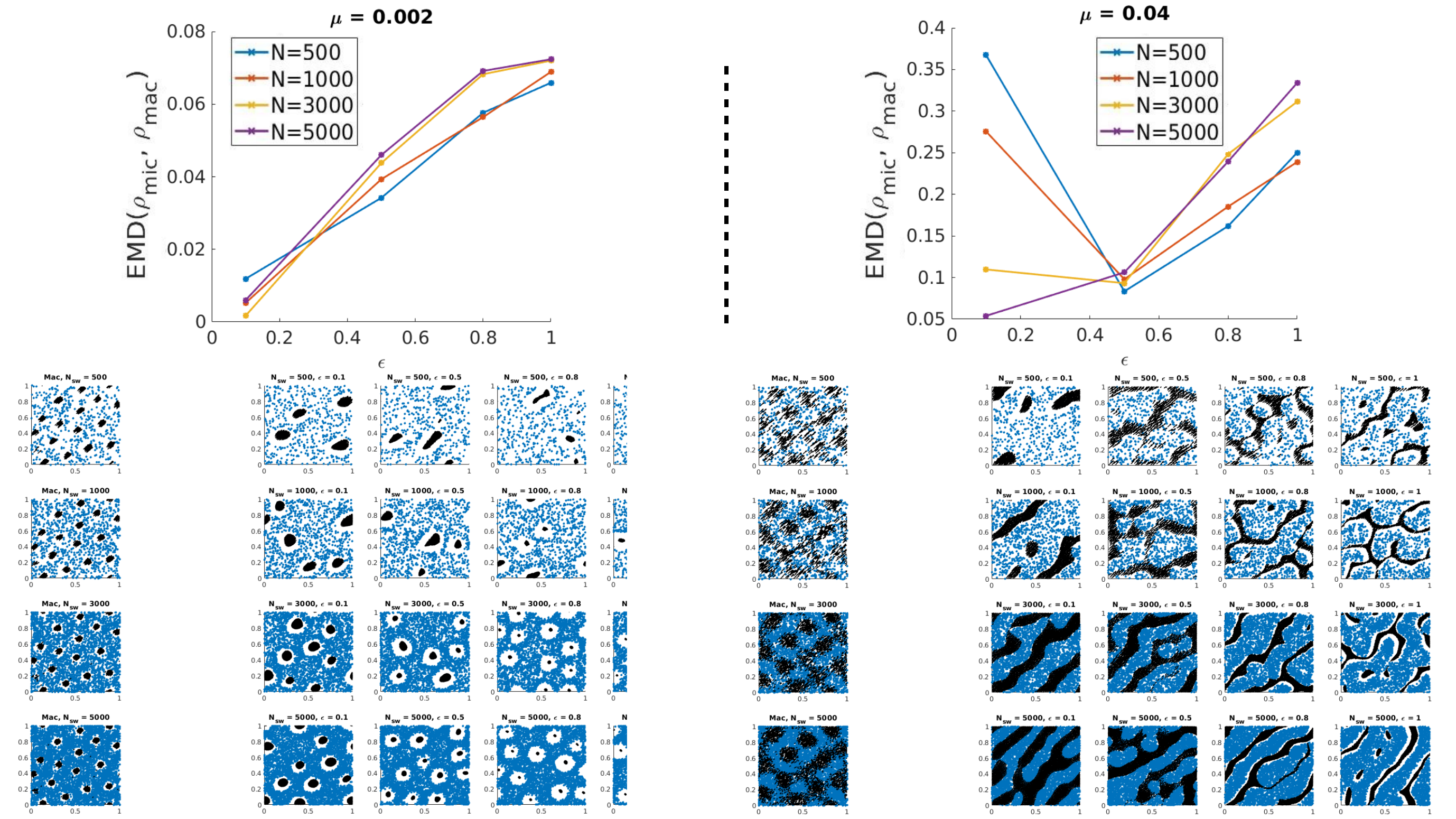}
\caption{Comparison between the discrete and continuum  simulations for mild obstacle spring stiffness $\kappa  = 100$ and agent friction $\zeta = 0.5$. Left figures: for weak agent agent repulsion $\mu = 2 \; 10^{-3}$, right figures, for $\mu = 4 \; 10^{-2}$. Top figures: EMD between the approximated continuum density and the discrete one as a function of $\epsilon$ for different values of the number of agents $N$: $N=500$ (blue curve), $N=1000$ (red curve) and $N=3000$ (yellow curve) and $N=5000$ (purple curve).  Bottom tables: simulations of the continuum model (left column), and of the discrete one for different values of $\epsilon$: $\epsilon = 0.1$ (second column), $\epsilon = 0.5$ (third column), $\epsilon = 0.8$ (fourth column) $\epsilon = 1$ (last column). The different rows correspond to different number of agents for the discrete simulations as well as for the discretization of the continuum density (from top to bottom: $N=500$, $N=1000$, $N=3000$, $N=5000$).\label{micmack100}}
\end{figure}

Fig. \ref{micmack100} suggests that  the discrete and continuum models are in quite good agreement in the case of week agent-agent repulsion ($b_p \ll 1$, left panel), where both models are able to reproduce agent clusters, while their correspondence is more tenuous for stronger agent-agent repulsion ($b_p$ close to the instability threshold, right panel), where the discrete dynamics seems to produce more trail-like patterns than the continuum model. 
For both regimes however, we can observe a significant improvement of the discrete-continuum correspondence as $\epsilon$ decreases, suggesting that the continuum model becomes a good approximation of the discrete dynamics as $\epsilon$ goes to zero. Indeed, for weak agent-agent repulsion (left panel), we observe that decreasing $\epsilon$ is accompanied by an increase in the cluster sizes and a decrease of the distance between the boundary of the clusters and the obstacles, getting closer to the cluster types observed with the macroscopic dynamics. For stronger agent-agent repulsion (right panel), the clusters thicken as $\epsilon$ decreases and get closer to the continuum structures.

These observations are confirmed by the measurements of the EMD between the discrete and continuum agent distributions (top plots of Fig. \ref{micmack100}). 
Indeed, one notes in the left panel that the distance between the two distributions decreases as the scaling parameter $\epsilon$ decreases, independently on the number of agents. Moreover, the top plot on the right panel shows that the discrete-continuum distance is larger for stronger agent-agent repulsion ($b_p$ close to 1) compared to the case where $b_p \ll 1$ (left panel). From the right figure, we also observe a strong dependency of the discrete-continuum distance as a function of the number of agents used in the discrete model. When the agent-agent repulsion is strong (or equivalently when $b_p$ is close to 1), it becomes crucial to use a large number of individuals for the discrete simulations, while the number of agents does not seem to significantly impact the discrete-continuum agreement in regimes favoring the apparition of small and dense clusters (small agent-agent repulsion or equivalently small $b_p$).  

These first observations tend to suggest that the choice of the number of agents in the discrete setting seem to depend both on the choice of $\epsilon$ and on the regime of parameters. In order to give more insights on the influence of $N$ and $b_p$ on the discrete-continuum match, we plot in Fig. \ref{micmack100error} the EMD between the discrete and continuum models as a function of $b_p$ (by changing the value of $\mu$ for fixed $\kappa=100$), having fixed $\epsilon = 0.1$ and for different $N$:$N=500$ (blue curve), $N=1000$ (red curve) and $N=3000$ (yellow curve) and $N=5000$ (purple curve). 

\begin{figure}
\centering 
\includegraphics[scale = 0.4]{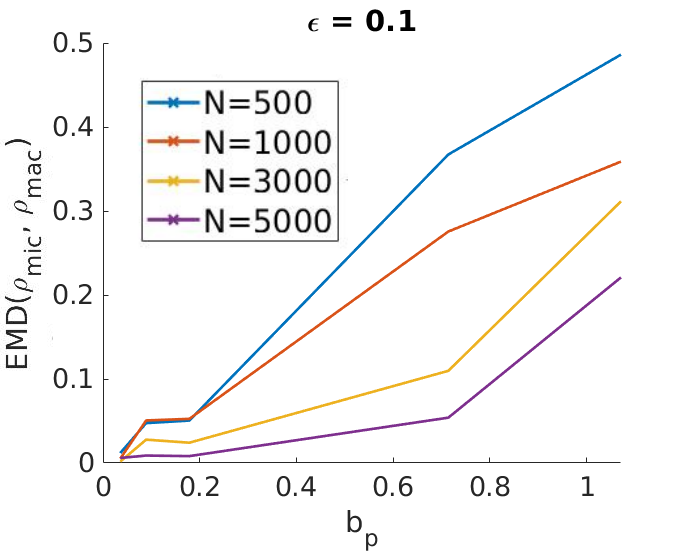}
\caption{EMD between the approximated continuum density and the discrete one as a function of $b_p$ for $\kappa=100$, $\zeta = 0.5$ and $\epsilon=0.1$, and for different values of the number of agents $N$: $N=500$ (blue curve), $N=1000$ (red curve) and $N=3000$ (yellow curve) and $N=5000$ (purple curve).  \label{micmack100error}}
\end{figure}

As one can see in Fig. \ref{micmack100error}, the discrete-continuum distance increases with $b_p$ independently on the number of agents $N$, suggesting indeed that the discrete and continuum models are closer far from the instability threshold. As the agent-agent repulsion increases (increasing values of $b_p$), the number of agents used in the discrete simulations has increasing influence on the match between the discrete and continuum simulations. These results suggest that large agent clusters with low density are better captured by a large number of agents. 

\subsubsection{Strong obstacle spring stiffness}
Here we aim to study the discrete-continuum agreement for strong obstacle spring stiffness $\kappa = 1000$. In Fig. \ref{micmack1000}, the left panel is obtained for $\mu = 2.10^{-4}$ (corresponding to a bifurcation parameter $b_p \approx 0.036$), and the right panel is for $\mu = 4.10^{-3}$ (corresponding to $b_p \approx 0.7$, close to the stability threshold 1). Top figures show the EMD between the continuum  and discrete solutions as a function of $\epsilon$, for different number of agents used for the discrete simulations $N$: $N=500$ (blue curve), $N=1000$ (red curve), $N=3000$ (yellow curve) and $N=5000$ (purple curves). As in the previous section, the corresponding simulations are shown below in tables: for each, the left column shows the simulations of the continuum model, and the next columns are simulations of the discrete model for different values of $\epsilon$: $\epsilon = 0.05$ (second column), $\epsilon = 0.1$ (third column), $\epsilon = 0.5$ (fourth column) $\epsilon = 0.8$ (fifth column) and $\epsilon = 1$ (last column). As before, the different rows of the tables correspond to different number of agents for the discrete simulations as well as for the discretization of the continuum density (from top to bottom: $N=500$, $N=1000$, $N=3000$, $N=5000$). 

\begin{figure}
\includegraphics[scale = 0.52]{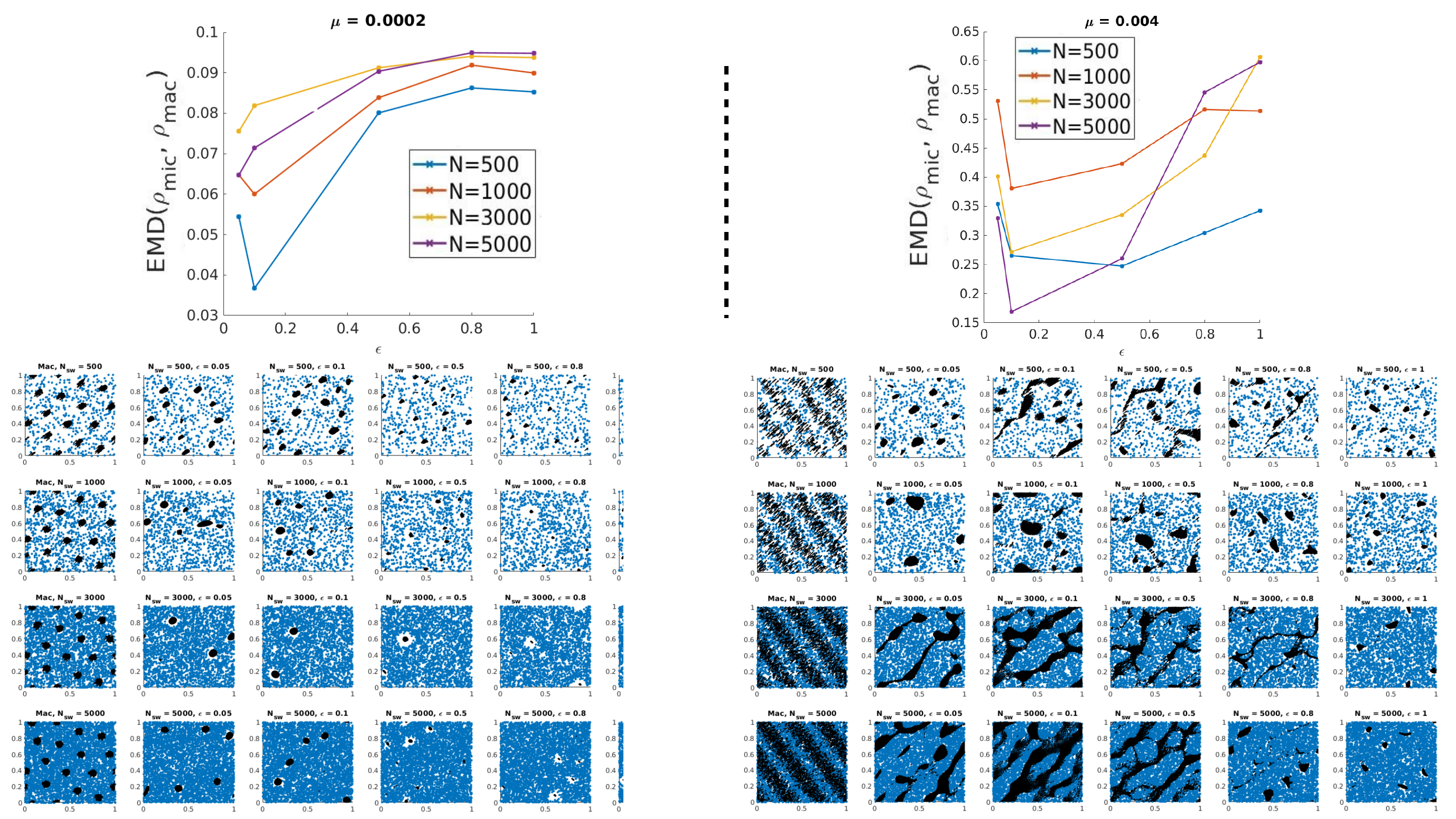}
\caption{Comparison between the discrete and continuum  simulations for strong obstacle spring stiffness $\kappa  = 1000$ and agent friction $\zeta = 0.5$. Left figures: for weak agent agent repulsion $\mu = 2 \; 10^{-4}$, right figures, for $\mu = 4 \; 10^{-3}$. Top figures: EMD between the approximated continuum density and the discrete one as function of $\epsilon$ for different values of the number of agents $N$: $N=500$ (blue curve), $N=1000$ (red curve) and $N=3000$ (yellow curve) and $N=5000$ (purple curve).  Bottom tables: simulations of the continuum model (left column), and of the discrete one for different values of $\epsilon$: $\epsilon = 0.05$ (second column), $\epsilon = 0.1$ (third column), $\epsilon = 0.5$ (fourth column) $\epsilon = 0.8$ (fifth column) and $\epsilon = 1$ (last column). The different rows correspond to different number of agents for the discrete simulations as well as for the discretization of the continuum density (from top to bottom: $N=500$, $N=1000$, $N=3000$, $N=5000$) \label{micmack1000}}
\end{figure}

For strong obstacle spring stiffness $\kappa = 1000$, we again observe that the discrete and continuous models are in good agreement far from the instability threshold (left panel), where both models reproduce clusters, while the agreement between the two models worsen for stronger agent-agent repulsion (right panel), where the discrete system fails to reproduce the travelling bands patterns observed with the continuum model. Again, the discrete-continuum agreement improves as $\epsilon$ decreases: for low agent-agent repulsion the discrete pattern sizes converge to the ones of the continuum model as $\epsilon$ decreases (from right to left in the left panel), and for strong agent-agent repulsion (last rows of the right panel), decreasing $\epsilon$ induces a phase transition between clustered states and trail-like agent patterns, closer to the formation of bands.

It is noteworthy that for small agent-agent-repulsion ($b_p \ll 1$, left figure), the agreement between the discrete and continuum dynamics seem to be better when using less agents in the discrete model, independently on the value of $\epsilon$ (compare purple and blue curves on the left panel), while close to the instability threshold ($b_p$ close to 1, right figure) the choice of $N$ seems to be related to $\epsilon$: the discrete-continuum error decreases when using larger $N$ for small $\epsilon$, smaller $N$ for larger $\epsilon$. 

\subsubsection{Summary of observations}
We conclude that the continuum equations are a good approximation of the discrete dynamics in the limit of small rescaling parameter $\eps$, as long as the agent-agent repulsion $\mu$ is small enough (\textit{i.e} in a parameter regime far from the instability threshold, $b_p \ll 1$). 
On the contrary, the trend is less apparent when $\mu$ gets closer to the instability threshold $\mu^*$ (corresponding to $b_p = 1$). 
In particular, when $\mu \approx \mu_*$, the rescaling factor $\eps$ can act as a phase transition parameter between different types of patterns (right panel, Fig. \ref{micmack1000}). This phase transition is due to the fact that the instability condition is given by \eqref{eq:instability_condition}.
Indeed, the presence of $\mu$ in this formula hints to the fact that at the discrete level the agent-agent repulsion potential $\psi$ plays a key role in determining the patterns that emerge. 
Therefore, it is no wonder that by rescaling the value of agent-agent repulsion radius $\tilde r_R=\eps r_R$ (and therefore changing the value of $\psi$) the shape of the patterns also changes. 
However, the smaller the $\mu$ the less relevant the role of $\psi$, thus the predictions of the continuum simulations become more robust.

There is also another important factor to take into account: the continuum dynamics just gives averaged behaviour of the discrete dynamics. 
If there is a wide variability in the discrete dynamics, due to its intrinsic stochasticity, then the average behaviour will not be able to represent well particular realizations of the discrete dynamics. 
It seems that closer to the boundary of the instability region ($\mu\approx \mu_*$) this variability is larger.

\section{Discussion}
In this article we have investigated a model for collective dynamics in an environment filled with obstacles that are tethered to a fixed point \emph{via} a spring. 
The model was first introduced in \cite{Aceves2020}. 
In particular, the paper has presented the following novelties:
\begin{inparaenum}[(i)]
    \item phase diagram of the continuum equations in dimension 2;
    \item a linear stability analysis of constant solutions;
    \item method to discriminate between different types of patterns that has been used to compare quantitatively the relation between discrete and continuum simulations;
    \item and, finally, a more extensive phase diagram of the discrete dynamics that has allowed to identify two new types of patterns with respect to \cite{Aceves2020} (honey comb structures and pinned cluster states).
\end{inparaenum}

\medskip
The continuum description captures well the behavior of the system when it is comprised of a large number of agents and obstacles, and involves huge computational savings compared with the simulation of the discrete system.
Comparing discrete and continuum simulations is in general not straightforward. We have proposed a method to compare the two types of solutions to investigate in which parameter regime they are in good correspondence. This parameter regime includes the assumptions made for the derivation of the continuum equations in Section \ref{sec:assumptions}: the spring stiffness must be large $\kappa\gg 1$, the number of agents and obstacles must be large $N,M\gg 1$, the scaling parameter $\eps\ll 1$ (see \eqref{eq:rescaling_parameters} for the rescaled parameters) must be small. However, we require one more condition to have a good correspondence between discrete and continuum dynamics: the agent-agent repulsion intensity $\mu$ must be much smaller than the critical value $\mu_*$, which is at the threshold of the instability condition \eqref{eq:instability_condition}. 
For values closer to $\mu_*$ the intrinsic variability of the system is too large to be described just with the averaged behaviour that captures the continuum equations.

\medskip
This work has also showcased the impact of the environment in pattern formation in collective dynamics. The phase diagrams of both discrete and continuum dynamics show that the feedback interactions between agents and obstacles give rise to a rich variety of patterns. 
In particular, we have observed that trails, travelling bands, moving clusters, uniform configurations and other in-between patterns emerge. 
The fact that agents can modify their environment by moving the obstacles is fundamental to this pattern emergence. 
This can be clearly seen in the linear stability analysis where the instability condition \eqref{eq:instability_condition} depends crucially on the agent-obstacle repulsion force $\phi$\ which is the only interaction force between agents and obstacles, and on the spring stiffness $\kappa$ which indicates the degree of mobility of the obstacles around their tethered positions. 

\medskip
As a prospective work, we would like to use the models investigated here to study the impact of the environment in collective dynamics under different set up. One of these set up is collective motion in a complex fluid. To investigate this, the idea is to couple the current model with a fluid model. Then the environment in which collective motion takes place will be the combination of the fluid with the obstacles. The idea of representing a complex fluid in this manner is similar to other existing models in the literature, such as the Oldroyd-B model that describes the visco-elasticity of fluids filled with spring dumbbells \cite{oldroyd1950formulation}. 
The coupling of the current discrete model with a fluid model will require a new derivation of the continuum equations and a new linear stability analysis to understand how the presence of the fluid impacts the dynamics and pattern formation. \\
Another extension of this work will investigate the impact in collective dynamics of an environment filled with a different type of obstacles (i.e., obstacles of a different nature than the ones considered in this work). For example, one can consider solid obstacles that are movable but that are not tethered or that have a particular shape (like elongated fibers).

\section*{Acknowledgements}
The authors wish to thank Prof. E. Keaveny, Department of Mathematics, Imperial College London, for stimulating discussions. PD holds a visiting professor association with the Department of Mathematics, Imperial College London, UK. Part of this research was done when PD was affliated with this Department and supported by the Engineering and Physical Sciences Research Council (EPSRC) under grant no. EP/P013651/1. SMA was supported by the Vienna Science  and  Technology  Fund  (WWTF)  with  a  Vienna  Research  Groups  for  Young  Investigators, grant VRG17-014. The research by SMA was partially supported by the Austrian Science Fund (FWF) through the  project  F65.

\bibliographystyle{abbrv}
\bibliography{biblio}
\clearpage

		\appendix
\section{Supplementary material: Videos of IBM simulations \label{AppendixA}}

In the following paragraphs, we give some details on the videos of the IBM simulations available online as supplementary material. Each video is composed of two simulations representative of the different types of patterns shown in Fig. \ref{simus_micro} and highlighted by a red cross. In each movie, agents are represented by black arrows and obstacles by colored points. The colors indicate the distance of the obstacles to their tethered points (from blue (close to their attachment site) to red (stretched springs)).
Otherwise stated, the parameters used for the simulations are the ones indicated in table \ref{table_param}.

\textit{S1 - Trails}

Link: https://doi.org/10.6084/m9.figshare.19615599.v2

In this video, two simulations are shown : The left movie is obtained for weak obstacle spring stiffness $\kappa = 10$, agent friction $\zeta = 2$ and  agent-agent repulsion $\mu = 0.02$ and the right movie is obtained for mild obstacle spring stiffness $\kappa = 100$, agent friction $\zeta = 1$ and agent-agent repulsion $\mu = 0.04$. Both simulations show the spontaneous formation of trails of agents in a more or less deformable field of obstacle. For weak obstacle spring stiffness (left movie), agents easily repulse the obstacles as they move, creating large trails empty of obstacles that can merge or evolve over time. For larger obstacle spring stiffness (right movie), the agents also organize in trails that push the obstacles, creating tunnels with strong walls, more robust over time. It is noteworthy that trail structures are only observed for weak or mild obstacle spring stiffness.

\textit{S2 - Travelling bands - strong obstacles}

Link: https://doi.org/10.6084/m9.figshare.19615884.v3

Here, both simulations feature strong obstacle spring stiffness $\kappa = 1000$ and agent friction $\zeta = 2$. The left movie is obtained for weak agent-agent repulsion $\mu = 2.10^{-5}$ while the right movie is obtained for $\mu = 10^{-3}$. In both situations, agents end up organizing in travelling bands on the long run, but we can observe a first phase when agents try to organize in trails. This suggest that the trail-like formation is only stable when agents have enough strength to push the obstacles as they move. Comparing the left and right movie, we also observe that larger agent-agent repulsion leads to larger clusters of agents.

\textit{S3 - Honneycomb structures}

Link: https://doi.org/10.6084/m9.figshare.19615116.v2

We consider here the case of weak obstacle spring stiffness $\kappa = 10$. The left movie is obtained for agent friction $\zeta = 0.2$ and mild agent-agent repulsion $\mu = 0.05$ while the right movie is obtained for $\zeta = 1$ and larger agent-agent repulsion $\mu = 0.1$. As one can observe, when the agent-agent repulsion is large enough in an easily deformable obstacle field, the agent phase wins over the obstacle phase, creating regularly spaced islands of obstacles in the form of honneycomb structures.  

\textit{S4 - Moving Clusters}

Link: https://doi.org/10.6084/m9.figshare.19615878.v2

Here, the left movie features weak obstacles $\kappa = 10$, agent friction $\zeta = 0.2$ and agent-agent repulsion $\mu = 0.002$, while the right movie is for $\kappa = 100$, $\zeta = 1$ and $\mu = 0.004$. When agent-agent repulsion is low enough, agents spontaneously organize into more or less round clusters surrounded by obstacles. This cluster formation happens very fast, and agent clusters then move more or less fast depending on their environment. Large and slow clusters of agents are observed in the left movie, where obstacles are very loose and agent-agent repulsion is a bit larger, while more numerous, smaller and more stable agent clusters are observed in the right movie (where agent-agent repulsion is a bit larger but obstacles are stronger).

\textit{S5 - Pinned Clusters}

Link: https://doi.org/10.6084/m9.figshare.19615866.v1

In these movies, we consider a small agent friction $\zeta = 0.2$ and agent-agent repulsion $\mu = 2. 10^{-4}$, for mild obstacle spring stiffness $\kappa = 100$ (left pannel) and for strong obstacles $\kappa = 1000$ (right pannel). We classify these clusters as 'pinned' as the agents organize very fast into small and highly concentrated clusters that do not move (right pannel) or move very slow in the case of mild obstacles. We can also observe merging of clusters in the case of mild obstacles (left video). In these extreme case, the force exerted by the stretched obstacles overcomes the other forces, preventing the agent clusters to move further from their position. 

\section{Supplementary material: Numerical codes for the Individual-Based Model and Continuous Model simulations \label{AppendixB}}

\textit{Individual-Based model code}

The MATLAB code corresponding to the simulations of Fig. 2 (Individual-Based model) can be found at the following link:
https://doi.org/10.6084/m9.figshare.19937861.v1

This supplementary material contains the following files: 
\begin{itemize}
	\item run\_SwOb\_IBM\_demo.m : file which allows to fix the parameter values (set to reproduce Fig 1 of the paper), create folders for registering the datafiles generated by the model, and run the simulations
	\item SwObIBM.m : MATLAB function that allows, given a set of parameters provided as entry, to run an entire simulation and register regularly the files in the specified data folder 
	\item The other files (get\_IC.m, get\_neighborhood\_info.m, get\_pushing\_forces.m, make\_plot.m) are intrinsic functions of the model, the description of which is contained in the corresponding files.
\end{itemize}

\textit{Continuous model code}

The MATLAB code corresponding to the simulations of Fig. 3 (Continuous model) can be found at the following link:
https://doi.org/10.6084/m9.figshare.19939409.v1

This supplementary material contains the following files: 
\begin{itemize}
	\item run\_Swimmer\_demo.m : file which allows to fix the parameter values (set to reproduce Fig 3 of the paper), create folders for registering the datafiles generated by the model, and run the simulations
	\item SwimmerSOH.m : MATLAB function that allows, given a set of parameters provided as entry, to run an entire simulation and register regularly the files in the specified data folder 
	\item The other files (get\_constants.m, get\_eigenvalues\_F.m, get\_eigenvalues\_G.m, get\_F.m, get\_G.m, get\_fluxes.m, get\_Jacobian\_F.m, get\_Jacobian\_G.m, mmat.m, convol.m, make\_plot.m) are intrinsic functions of the model, the description of which is contained in the corresponding files.
\end{itemize}

\section{Numerical method for the continuum  model} \label{num_method}
One of the difficulties in solving the nonlinear model \eqref{eqmacro} is the geometric constraint $|\Omega| = 1$, and the resulting non-conservativity of the model arising from
the presence of the projection operator $P_{\Omega^\perp}$. We rely on a method proposed in \cite{motsch2011numerical,degonddimarco2015} where the SOH model is approximated by a relaxation problem consisting of an unconstrained conservative hyperbolic system supplemented with a relaxation operator onto vector fields satisfying the constraint $|\Omega| = 1$. The relaxation model writes:
\begin{align}
	\partial_t \rho_g^\epsilon &+ \nabla \cdot \big(U^\epsilon \rho_g^\epsilon\big) = 0 \nonumber \\
	\partial_t ( \rho_g^\epsilon \Omega^\epsilon) &+ \nabla \cdot \bigg( \rho_g^\epsilon V^\epsilon \otimes \Omega^\epsilon \bigg) + d_3  \nabla \rho_g^\epsilon - \gamma_s  \Delta \big( \rho_g^\epsilon \Omega^\epsilon \big) =  \frac{\rho_g^\epsilon}{\epsilon} (1-|\Omega^\epsilon|^2)\Omega^\epsilon \label{relaxmac},
\end{align}
\noindent where
\begin{align*}
	U^\epsilon &= d_1 \Omega^\epsilon - \frac{1}{\xi } \nabla \bar{\rho^\epsilon}_f - \frac{\mu}{\xi} \nabla \rho^\epsilon_g, \\
	V^\epsilon  &= d_2 \Omega^\epsilon - \frac{1}{\xi } \nabla \bar{\rho}^\epsilon_f - \frac{\mu}{\xi} \nabla \rho^\epsilon_g.
\end{align*}

In the limit $\epsilon \rightarrow 0$, one can shows that \eqref{relaxmac} converges towards \eqref{eqmacro}. The main idea is based on the fact that the right-hand side of \eqref{relaxmac} is parallel to $\Omega^\epsilon$, and the proof is similar to \cite{motsch2011numerical}. 

The numerical method is adapted from \cite{motsch2011numerical}, using the so-called splitting scheme in two steps (dropping the $\epsilon$ for clarity):
\begin{itemize}
	\item Step 1: Solve the conservative part: 
	\begin{align}
		\partial_t \rho_g &+ \nabla \cdot \big(U \rho_g\big) = 0  \label{conserv}\\
		\partial_t ( \rho_g \Omega) &+ \nabla \cdot \bigg( \rho_g V \otimes \Omega \bigg) + d_3  \nabla \rho_g - \gamma_s  \Delta \big( \rho_g \Omega \big) = 0,
	\end{align}
	\item Solve the relaxation part:
	\begin{align*}
		\partial_t \rho_g &= 0 \nonumber \\
		\partial_t ( \rho_g \Omega) &= \frac{\rho_g}{\epsilon} (1-|\Omega|^2)\Omega
	\end{align*}
\end{itemize}

Writting $Q = \begin{pmatrix}  \rho_g \\ \rho_g \Omega_x \\ \rho_g \Omega_y \end{pmatrix}$,  and $\tilde{Q} = \rho_g \ast \Phi \ast \Phi$, system \eqref{conserv} can be written as:
\begin{equation}\label{conservativeform}
	\partial_t Q + \partial_x F(Q,Q_x,\Delta Q_x) + \partial_y G(Q,Q_y,\Delta Q_y) = 0,
\end{equation}
where we have denoted by $\Delta Q_x$ (resp. $\Delta Q_y$) the derivative in $x$ (resp. in $y$) of the Laplacian of $Q$, and the fluxes write
$$
F(Q,Q_x, \Delta \tilde{Q}_x) = \begin{pmatrix} d_1 Q(2) - \frac{1}{\xi} \frac{\gamma \rho_A}{\eta} Q(1) \Delta \tilde{Q}_x \ast \Phi \ast \Phi -  \frac{\mu}{\xi} Q(1) Q_x(1)\\
	d_2 \frac{Q(2)^2}{Q(1)} - \frac{1}{\xi} \frac{\gamma \rho_A}{\eta} Q(2) \Delta \tilde{Q}_x \ast \Phi \ast \Phi  - \frac{\mu}{\xi} Q(2) Q_x(1) + d_3 Q(1) - \gamma Q_x(2)\\
	d_2 \frac{Q(2)Q(3)}{Q(1)} - \frac{1}{\xi} \frac{\gamma \rho_A}{\eta} Q(3) \Delta \tilde{Q}_x \ast \Phi \ast \Phi  - \frac{\mu}{\xi} Q(3) Q_x(1)  - \gamma Q_x(3),
\end{pmatrix}
$$
and 
$$
G(Q,Q_y, \Delta \tilde{Q}_y) = \begin{pmatrix} d_1 Q(3) - \frac{1}{\xi} \frac{\gamma \rho_A}{\eta} Q(1) \Delta \tilde{Q}_y \ast \Phi \ast \Phi -  \frac{\mu}{\xi} Q(1) Q_y(1)\\
	d_2 \frac{Q(2)Q(3)}{Q(1)} - \frac{1}{\xi} \frac{\gamma \rho_A}{\eta} Q(2) \Delta \tilde{Q}_y \ast \Phi \ast \Phi  - \frac{\mu}{\xi} Q(2) Q_y(1) - \gamma Q_y(2) \\
	d_2 \frac{Q(3)^2}{Q(1)} - \frac{1}{\xi} \frac{\gamma \rho_A}{\eta} Q(3) \Delta \tilde{Q}_y \ast \Phi \ast \Phi  - \frac{\mu}{\xi} Q(3) Q_y(1) + d_3 Q(1) - \gamma Q_y(3).
\end{pmatrix}
$$

The explicit time-discretization for Eq. \eqref{eqmacro} writes:
$$
Q^*_{i,j} = Q_{i,j}^n  - \frac{\Delta t}{\Delta x} \bigg( F^n_{i+\frac{1}{2},j} - F^n_{i-\frac{1}{2},j}\bigg) - \frac{\Delta t}{\Delta y} \bigg( G^n_{i,j+\frac{1}{2}} - G^n_{i,j-\frac{1}{2}}\bigg),
$$
where
$$
F_{i+\frac{1}{2},j} = \frac{F(Q_{i,j}) + F(Q_{i+1,j})}{2} - \frac{1}{2} P^2(\frac{\partial F}{\partial Q}(\bar{Q}_{i,j}, \bar{Q_x}_{i,j}, \bar{\Delta \tilde{Q}_x}_{i,j}) \big( Q_{i+1,j} - Q_{i,j}\big),
$$
with
$$
\bar{Q}_{i,j} = \frac{Q_{i,j}+Q_{i+1,j}}{2}, \; {Q_x}_{i,j} = \frac{Q_{i+1,j} - Q_{i,j}}{2}, \; \bar{Q_x}_{i,j} = \frac{{Q_x}_{i,j}+{Q_x}_{i+1,j}}{2}, 
$$
$$
\Delta \tilde{Q}_{i,j} \approx \frac{\tilde{Q}_{i-1,j} + \tilde{Q}_{i+1,j} - 2 \tilde{Q}_{i,j}}{\Delta x^2} + \frac{\tilde{Q}_{i,j-1} + \tilde{Q}_{i,j+1} - 2 \tilde{Q}_{i,j}}{\Delta y^2}
$$
and where $P^2(\frac{\partial F}{\partial Q})$ is a second polynomial of a matrix at the intermediate state of  $(\bar{Q}_{i,j}, \bar{Q_x}_{i,j}, \bar{\Delta \tilde{Q}_x}_{i,j}) $ and $(\bar{Q}_{i+1,j}, \bar{Q_x}_{i+1,j}, \bar{\Delta Q_x}_{i+1,j}) $ computed following \cite{degond1999polynomial}. Terms in $G$ are computed the same way. 

Convoluted terms are computed using a double fast-Fourier transform as:
$$
\tilde{Q} = \rho \ast \Phi \ast \Phi = F^{-1} \bigg[ F\bigg( F^{-1} \big( \hat{ \rho} \hat{\Phi}\big)  \hat{\Phi}\bigg)  \bigg],
$$
\noindent where we used the fast Fourier transform.


\end{document}